\documentclass{svproc}
\usepackage[a4paper, margin=2.5cm]{geometry} 
\usepackage{lipsum} 
\renewenvironment{proof}[1][Proof]{\par\noindent\textbf{#1.} }{\hfill$\square$\par\vspace{1em}}

\usepackage{url}
\usepackage{xcolor}
\usepackage{mathtools}
\usepackage[normalem]{ulem}
\usepackage{enumitem}
\usepackage{subcaption}

\usepackage{array}
\usepackage[hidelinks]{hyperref}
\usepackage{graphicx}
\usepackage{epstopdf}
\usepackage{tikz}
\usepackage{pgfplots} 
\pgfplotsset{compat=newest}
\usetikzlibrary{external}
\usepackage{amsmath}
\usepackage{amsfonts}
\usepackage{amssymb}
\usepackage{bm}
\usepackage{bbm}        
\usepackage[dvipsnames]{xcolor}
\usepackage{cite}

\newcommand{\realnumbers}{\mathbb{R}}
\newcommand{\realpositive}{\realnumbers_0^+}

\newcommand{\dsomething}[1]{\mathrm{d}#1}

\newcommand{\RLintegralInput}[2]{\operatorname{I}_{#2}^{#1}}
\newcommand{\RLderivativeInput}[2]{{}^\mathrm{RL}\operatorname{D}_{#2}^{#1}}
\newcommand{\CderivativeInput}[2]{{}^\mathrm{C}\operatorname{D}_{#2}^{#1}}
\newcommand{\AstderivativeInput}[2]{\operatorname{D}_{\ast #2}^{#1}}

\newcommand{\pVec}{\mathbf p}
\newcommand{\hMat}{\mathbf H}
\newcommand{\HillSoln}{y_\mathrm{Hill}}
\newcommand{\simSoln}{y_\mathrm{sim}}
\newcommand{\assp}{(A2)}
\newcommand{\assJ}{(A1)}
\newcommand{\epsC}{\varepsilon}

\begin{document}
	\mainmatter              
	%
%
\title{Hill-Type Stability Analysis of Periodic Solutions of Fractional-Order Differential Equations} 

\titlerunning{Stability of Periodic Solutions of FODEs}

\author{
	Paul-Erik Haacker\inst{1},
	Remco I. Leine\inst{1},
	Renu Chaudhary\inst{2},
	Kai Diethelm\inst{2},
	Andr\'e Schmidt\inst{1},
	Safoura Hashemishahraki\inst{2}
}

\authorrunning{Paul-Erik Haacker et al.}

\tocauthor{Paul-Erik Haacker, Remco I. Leine, Renu Chaudhary, Kai Diethelm, Andr\'e Schmidt, Safoura Hashemishahraki}

\institute{
	Institute for Nonlinear Mechanics, University of Stuttgart, 70569 Stuttgart, Germany\\
	\email{haacker@inm.uni-stuttgart.de},\\
	\texttt{https://www.inm.uni-stuttgart.de/en/}
	\and
	Faculty of Applied Natural Sciences and Humanities, Technical University of Applied Sciences Würzburg-Schweinfurt, 
	97421 Schweinfurt, Germany
}

	\maketitle              
	
	\begin{abstract}{This paper explores stability properties of {periodic solutions of} (nonlinear) fractional-order differential equations (FODEs). { As classical Caputo-type FODEs do not admit exactly periodic solutions, we propose a framework of Liouville-Weyl-type FODEs, which do admit exactly periodic solutions and are an extension of Caputo-type FODEs. }Local linearization around a periodic solution results in perturbation dynamics governed by a linear time-periodic differential equation. In the classical integer-order case, the perturbation dynamics is therefore described by Floquet theory, i.e. the exponential growth or decay of perturbations is expressed by Floquet exponents which can be assessed using the Hill matrix approach. For fractional-order systems, however, a rigorous Floquet theory is lacking. Here, we explore the { limitations} when trying  to extend Floquet theory and the Hill matrix method to linear time-periodic fractional-order differential equations (LTP-FODEs) as local linearization of nonlinear fractional-order systems. A key result of the paper is that such an extended Floquet theory can only assess exponentially growing solutions of LTP-FODEs. Moreover, we provide an analysis of linear time-invariant fractional-order systems (LTI-FODEs) with algebraically decaying solutions and show that the inaccessibility of decaying solutions through Floquet theory is already present in the time-invariant case.
}
		\keywords{ fractional derivative, fractional-order differential equation, (linear) time-periodic systems, Floquet theory, stability, Hill problem.}
	\end{abstract}
	{\smash{Notation:}}
	\begin{itemize}[label=$\cdot$]
		\item $\realpositive$ is the set of real, positive numbers.
        \item {For a tuple $x = (x_1, \dots,x_n) \in \realnumbers^n$, $||x||_2 = \sqrt{\sum_{k=1}^n x_k^2}$ denotes the 2-norm.}
		\item For a function $f:\realnumbers\to\realnumbers$, $f'$ denotes its first derivative.
		\item For a function $f:\realnumbers\supseteq M \to \realnumbers$, $||f||_\infty := \sup_{x\in M} |f(x)|$.
		\item $L^1[a,b]$ denotes the space of functions $f:[a,b]\to\realnumbers$ with $\int_a^b|x(t)|~\dsomething{t}<\infty$.
		\item $L^\infty[a,b]$ denotes the space of functions $f:[a,b]\to\realnumbers$ with $||f||_\infty<\infty$.
		\item $A^n[a,b]$ denotes the space of functions $f:[a,b]\to \realnumbers$ which are $(n-1)$-times differentiable and the $(n-1)$th derivative is absolutely continuous on $[a,b]$.
        \item {$AC[a,b]$ denotes the space of functions $f:[a,b]\to \realnumbers$ which are absolutely continuous on $[a,b]$. We say that a function $f$ is locally absolutely continuous on $(-\infty,0]$, if it is absolutely continuous on every finite subinterval. We denote this $f\in AC_\mathrm{loc}(-\infty,0]$.}
	\end{itemize}
	\section{Introduction}
	{We investigate the stability of periodic solutions of nonlinear fractional-order differential equations, focusing on the challenges of extending Floquet theory to linear time-periodic fractional-order differential equations (LTP-FODEs). We show that while a Hill-type approach can capture exponentially growing solutions, it fails to account for decaying ones. This limitation is traced back to linear time-invariant fractional-order systems, where solutions typically decay algebraically rather than exponentially.
	}
	
	{
		Recently, interest in fractional-order differential equations has grown in various fields. A number of dynamical processes with memory have proven to be better described using fractional-order differential equations, where integer-order models fail to provide low-degree-of-freedom representations. Among others, applications include viscoelastic materials, diffusion processes and control theory \cite{Das2011}. Such – possibly nonlinear – systems admit periodic solutions when periodically forced or in the presence of self-excitation. Common formulations of Caputo differential equations cannot have exactly periodic solutions \cite{Tavazoei2009}, which is closely related to the fact that Caputo derivatives and Riemann-Liouville integrals of periodic functions are only asymptotically periodic \cite{Area2016}. Studies \cite{Rand2010,Daibiri2015} concerned with finding steady-state solutions of differential equations involving Liouville-Weyl-type fractional-order derivatives have successfully applied harmonic balance methods. A linearization of the system dynamics about a periodic solution, describing the dynamics of small perturbations, yields a LTP-FODE. This motivates the study of the stability of such systems.
		
		Lyapunov-stability of equilibria and $\Omega$-limit sets is a central concept in the study of dynamical systems and has also been studied for fractional order systems. Classically, stability of linear time-invariant fractional-order differential equations can be studied through Mittag-Leffler functions, see \cite{Diethelm2010} and the references cited therein. As fractional-order systems typically admit algebraic decay \cite{CongEtAl2020}, the notion of Mittag-Leffler stability was introduced analogously to exponential stability \cite{LiEtAl2009}. Furthermore, as fractional-order derivatives are nonlocal, FODEs are similar to functional differential equations \cite{Hale1977}. This similarity has been exploited by \cite{Hinze2020,Hinze2020a} in mechanical systems with fractional-order damping, treating the system as a functional differential equation with infinite delay and finding Lyapunov-functionals through the infinite-state representation (also called diffusive representation or kernel compression), which is often used for efficient numerical approximations of fractional-order integrals \cite{Chaudhary2024}. LTI-FODEs perturbed by small nonlinearities were studied in \cite{CongEtAl2018}. That work also analyzes linear time-varying FODEs, but only under the assumption of sufficiently small time-dependent terms. To the best of our knowledge, no stronger results on stability of solutions of linear time-variant fractional-order Caputo differential equations are available.
		
		In classical integer-order models (i.e. first-order ordinary differential equations), stability of periodic solutions may be analyzed via the eigenvalues of the infinite-dimensional Hill matrix which correspond to the Floquet exponents, see \cite{Bayer2023a} and the references cited therein. These problems are understood quite deeply, with Floquet theory as their foundation. This theory may be applied in bifurcation analysis of periodic motion of mechanical systems, which is self-excited \cite{Horvath2022} or externally forced \cite{Youssef2025}. { For LTP-FODEs, no similar theory has been established, although several attempts for the development of a Floquet theory for fractional-order systems have been undertaken \cite{Rezazadeh2016,Iomin2023}.  In particular, the work \cite{Rezazadeh2016} assumes that the Mittag-Leffler function approximately satisfies a multiplicative identity, despite the fact that the associated error may become arbitrarily large. Further, in \cite{Rezazadeh2016} a theory entirely based on approximations is developed, which only holds exactly in the first-order case ($\alpha =1 $). However, parameter identification schemes applied to experimental data typically yield fractional orders between $0.1$ and $0.7$ \cite{Pawlak2021,HinzeExperiment,Kim2009,Stark2021}.} {The work \cite{Iomin2023} considers a Schrödinger LTP-FODE of Caputo-type. Therein, a solution form expressed by a series of Mittag-Leffler functions is considered; however, an analysis of asymptotic properties of the solution form as well as a Hill-type method to access characteristic multipliers is missing. Furthermore, the proposed form of solutions of \cite{Iomin2023} uses a basis of Mittag–Leffler functions whereas the time-periodic system matrix (the Hamiltonian in the terminology of \cite{Iomin2023}) is represented in harmonic form. The use of two different series impedes a Hill-type approach, since the Mittag-Leffler function and the exponential function do not share a multiplicative identity.}
        
        To reach a Floquet theory, results on fundamental solution operators for LTP systems are lacking. The more general case of linear time-varying systems was considered by \cite{Gomoyunov2020}, developing a fundamental solution matrix. Defining a monodromy operator and analyzing its spectrum, as done for functional differential equations \cite{Hale1977}, remains an open challenge in the study of LTP-FODEs. Previous studies performed transition curve analysis on the fractional-order Mathieu equation \cite{Rand2010,Shahroudi2024}, as well as on more general forms ~\cite{Daibiri2015}, also including first-order derivatives and time-delays. The latter work employs fractional-order derivatives with infinite memory, i.e. Liouville-Weyl-type derivatives, which is of particular relevance to our study.  However, the analysis of transition curves in the parameter space only gives information on stability boundaries and thus fails to completely characterize stability over a parameter space, nor does it give information about the type of bifurcation which may originate at the stability boundary. 

         The literature review shows, that the classical Caputo-type differential equations do not admit exactly periodic solutions, but do admit them as limit sets. In order to study the stability properties of such limit sets, one must first formulate an extension of the class of Caputo differential equations to also include periodic solutions. The Liouville-Weyl derivative has these properties and a rigorous formulation of an initial value problem along with the set of initial conditions has to be developed. This then opens paths to extend classical theories on periodic solutions to fractional-order systems. For example, fundamental tools such as the Poincaré--Bendixson Theorem \cite{HirschEtAl2013}, which provides conditions for the existence of periodic solutions, need to be developed in the fractional-order case. Further, when given a periodic solution, possibly obtained using the harmonic balance method \cite{Rand2010,Daibiri2015,Shahroudi2024}, linearly perturbed dynamics can be studied to assess stability in the  hyperbolic case. Therefore, a Floquet-type theory on LTP-FODEs needs to be developed and the existence of Floquet exponents needs to be studied. For application of this theory to numerical stability analysis, methods like the Hill method should be extended to the fractional-order case.
        
		The aim of this paper is to study how Floquet theory can be extended to LTP-FODEs and how their Floquet exponents can be computed. Therefore, we first establish a formulation of a suitable initial value problem, which does admit periodic solutions. In order to understand the potential of Floquet theory in LTP systems, we study the subset of LTI systems, where explicit solutions may be constructed. 
        
        The main merit of this paper is the insight that exponentially decaying Floquet solutions (i.e. the ones related to asymptotically stable equilibria or stable manifolds of saddle equilibria) do not exist for LTP-FODEs, but exponentially growing Floquet solutions may be observed and assessed via a Hill-type approach. Parts of the research of this work were presented in a preliminary form {in reference} \cite{Haacker2025}, namely the fractional Hill problem was presented as well as results on linear time-invariant FODEs. The novel contributions of the present work are the formulation of the system class in Section \ref{sec_FODEs}, the analysis of decaying solutions in Section \ref{sec_LTI} as well as the proof of Theorem \ref{thm_fracHill} and the following Propositions.
	}
	
	The present paper is organized as follows: Mathematical preliminaries from fractional calculus as well as Hill's method for first-order LTP systems are briefly presented in  Section \ref{sec_preliminaries}. Section \ref{sec_FODEs} introduces the system class studied in the present paper, providing a framework suitable to study periodic solutions. The subclass of LTI systems is considered in Section \ref{sec_LTI}, where diverging solutions can be assessed though an eigenvalue problem and decaying solutions are studied by other means. A similar phenomenon appears in LTP systems in Section \ref{sec_LTP}. We construct a Hill-type problem and are able to trace decaying trajectories. Finally, we give an outlook on open problems in Section \ref{sec_outlook}.
	
	\section{Preliminaries}
	\label{sec_preliminaries}
	\subsection{Fractional Calculus }

	{
		We briefly recall some concepts of fractional calculus and refer the reader to the books \cite{Diethelm2010,Kilbas2006} for further details. Several definitions of fractional derivatives exist, each suited to different mathematical and physical contexts. Given a function $f \in L^1[a,b]$, the Riemann--Liouville (RL) fractional integral of order $\alpha > 0$ is defined as
		\begin{align}
			(I_{a}^\alpha f)(t) = \frac{1}{\Gamma(\alpha)} \int_a^t (t - \tau)^{\alpha - 1} f(\tau)\,d\tau,
		\end{align}
		where $\Gamma(\cdot)$ is the Gamma function. { In this work, we consider derivatives of fractional orders $\alpha \in (0,1)$.} The RL fractional derivative of order $\alpha \in (0,1)$ is given by
		\begin{align}
			\left(\RLderivativeInput{\alpha}{a} f\right)(t) = \left(\operatorname D ^1 {I_{a}^{1 - \alpha}} f\right)(t).
		\end{align} The RL derivative requires fractional-order initial conditions, which limits its practical applicability in initial value problems.
		
		To overcome this limitation of the RL derivative, the Caputo derivative is often used in applied settings \cite{Hinze2020,Rand2010,Daibiri2015}. For $\alpha \in (0, 1)$  and a function $f \in AC[a,b]$, it is defined as
		\begin{align}
			\left(\CderivativeInput{\alpha}{a} f\right)(t) = \left( {I_{a}^{1 - \alpha}}\operatorname D ^1 f\right)(t)= \frac{1}{\Gamma(1 - \alpha)} \int_a^t (t - \tau)^{ - \alpha} f'(\tau)\,d\tau,
		\end{align}
		and has the advantage that it allows for standard (integer-order) initial conditions. In the context of systems of differential equations, the fractional derivative is typically applied component-wise. For a vector-valued function $x(t) \in \mathbb{R}^n$, the fractional derivative is understood as
		\begin{align}
			\operatorname D^\alpha x(t) := \left(\operatorname D^\alpha x_1(t), \operatorname D^\alpha x_2(t), \dots, \operatorname D^\alpha x_n(t)\right)^\top,
		\end{align}
		with each $x_i$ differentiated individually.

        { A different formulation (see \cite{Diethelm2010} and the references therein) of the Caputo derivative reads
        \begin{align}
            \left(\AstderivativeInput{\alpha}{a} f\right)(t) = \left(\RLderivativeInput{\alpha}{a} [f-f(0)]\right)(t),
        \end{align} for the case of $\alpha \in (0,1)$. If $f\in AC[a,b]$, then the two formulations coincide $\AstderivativeInput{\alpha}{a}f = \CderivativeInput{\alpha}{a}f$ almost everywhere \cite{Diethelm2010}.}
		
		The solutions to linear fractional-order differential equations often involve the Mittag-Leffler function, which generalizes the exponential function. The two- and one-parameter Mittag-Leffler matrix function for $\alpha,\beta>0$ and $Z \in \mathbb C ^{m\times m}$ are defined by
		\begin{align}
			E_{\alpha,\beta}(Z) = \sum_{k=0}^{\infty} \frac{Z^k}{\Gamma(\alpha k + \beta)},\quad  E_\alpha(Z) =E_{\alpha,1}(Z),
		\end{align}
        respectively. In the case $\alpha = \beta= 1$, the Mittag-Leffler function reduces to the exponential function, i.e., $E_1(z) = e^z$.
		
		Finally, the Liouville--Weyl derivative of order $\alpha > 0$ is defined as
		\begin{align}
			(\CderivativeInput{\alpha}{-\infty} f)(t) = \frac{1}{\Gamma(1 - \alpha)}  \int_{-\infty}^t (t - \tau)^{- \alpha} f'(\tau)\,d\tau.
		\end{align} Here, we use the Caputo-type definition $\CderivativeInput{\alpha}{a} \to \CderivativeInput{\alpha}{-\infty}$ for $a \to -\infty$. One may show for sufficient assumptions on integrability and differentiability of $f$, that it holds that $\CderivativeInput{\alpha}{-\infty}f = \RLderivativeInput{\alpha}{-\infty}f$. 

	}
\subsection{{The Hill Method for First-Order Systems}}
{ Hill's method is a frequency-based approach to assess the stability of periodic solutions \cite{Zhou2002}. It has been applied to first order ordinary differential equations as well as delay differential equations \cite{Stepan2011}. Here, we summarize the approach for ordinary differential equations, which we develop for FODEs in Section \ref{sec_LTP}. For time-periodic first-order systems
	\begin{align}
		\dot \xi = \tilde f(t,\xi(t)),
	\end{align}
	where $\tilde f(t,\cdot) = \tilde f(t+T,\cdot)$ for all $t$ with some $T>0$, there may exist a known periodic solution $\xi_p(t) = \xi_p(t+T)$ and the desire to assess its stability properties. To do so, one can introduce the perturbation $\delta(t) = \xi(t)-\xi_p(t)$ and find its linearized dynamics to be
	\begin{align}
		\dot \delta (t) = \tilde J (t) \delta (t),
	\end{align}
	where $\tilde J(t) = \frac{\partial \tilde f}{\partial \xi}(t,\xi_p(t))$ is $T$-periodic. Assuming $\tilde J$ to be available in Fourier form and knowing that, by Floquet theory, the solution is of the form
	\begin{align}
		\delta(t) = \exp(\mu t)\tilde p(t),
	\end{align}
	where $\mu$ is called the Floquet exponent and $\tilde p$ is $T$-periodic, the Floquet exponents are eigenvalues of the infinite-dimensional Hill-matrix
	\begin{align}\label{eq_Hillmat}
		\hMat_\infty &= \begin{pmatrix}
			\ddots & \vdots & \vdots & \vdots & \\
			\dots & \tilde J_0 + i\omega I & \tilde J_{-1} & \tilde J_{-2} & \dots\\
			\dots & \tilde J_1 & \tilde J_0  & \tilde J_{-1} & \dots\\
			\dots & \tilde J_{2} & \tilde J_{1} & \tilde J_0 -i\omega I & \dots\\
			& \vdots & \vdots & \vdots & \ddots
		\end{pmatrix}.
	\end{align}
	Herein, $\tilde J_k$ are the Fourier coefficients of $\tilde J(t) = \sum_{k=-\infty}^{+\infty}\tilde J_k \exp(ik\omega t)$, $\omega = 2\pi / T$ and $I$ is the identity matrix of appropriate dimension. After truncation to finite dimension $N$, one may either compute all eigenvalues of the truncated Hill-matrix $\hMat_N$ and employ a sorting technique to extract Floquet exponents or perform the recently introduced Koopman-Hill projection \cite{Bayer2023a}, avoiding sorting and directly approximating Floquet-multipliers.}

\section{Liouville-Weyl Differential Equations: A Class of Systems with Periodic Solutions} \label{sec_FODEs}
We call
\begin{subequations}\label{eq_nonlinearIVP}
	\begin{align}\label{eq_nonlinearFODE}
		\CderivativeInput{\alpha}{-\infty}x(t) = f(t,x(t)),\quad t \geq t_0
	\end{align}
	a nonlinear Liouville-Weyl differential equation of fractional differentiation order $\alpha \in (0,1)$, where the function $f:[t_0,\infty)\times\realnumbers^n\to\realnumbers^n$ is required to be at least locally Lipschitz in its second argument and continuous in the first argument.
	The differential equation \eqref{eq_nonlinearFODE} is equipped with the initial condition
	\begin{align} \label{eq_nonlinearFODE_IC}
		x(t) = x_0(t),\quad t\leq t_0
	\end{align}
\end{subequations} 
	to form an initial value problem \eqref{eq_nonlinearIVP}. { For some $t_1>t_0$, we call a function $x:(-\infty,t_1) \to \realnumbers^n$ a \emph{strong solution} of the initial value problem \eqref{eq_nonlinearIVP} if it is continuous, coincides with the initial condition $x_0$ on $(-\infty,t_0]$ and satisfies the differential equation \eqref{eq_nonlinearFODE} for $t_0 < t < t_1$.} The initialization function $x_0$ is assumed to be an element of the function space $X(-\infty,t_0]$ given in Definition \ref{def_InitialSpace} below. Therein, the local boundedness of $x_0'$ at $t_0$ is indeed necessary, as we will see in Example \ref{ex_initialFunctions} \ref{item_unbouded_der_at0} below.
\begin{definition}\label{def_InitialSpace}
		The space of initial functions is defined as 
			\begin{align}
				X(-\infty,t_0] := &\Bigl\{ x_0 \in {AC_\mathrm{loc}(-\infty,t_0]} \cap L^\infty (-\infty,t_0]~\vert~ \exists \,\eta > 0: x_0' \in L^\infty [t_0-\eta,t_0]\Bigr\}
			\end{align}
		where {$AC_\mathrm{loc}(a,b]$} and $L^\infty(a,b]$ denote the function spaces of { locally absolutely continuous functions} and functions with a finite supremum-norm, respectively, on the interval $(a,b]$. 
	\end{definition}
    In order to {assess} the existence and uniqueness of solutions of problem \eqref{eq_nonlinearIVP}, we will split the Liouville-Weyl derivative 
    \begin{align} \label{eq_SplitDerivative1}
	\CderivativeInput{\alpha}{-\infty}x(t)= \int_{-\infty}^t\frac{(t-\tau)^{-\alpha}}{\Gamma(1-\alpha)}x'(\tau)\dsomething{\tau} = \mathcal{F}x_0(t) + \CderivativeInput{\alpha}{t_0}x(t).
\end{align}
into the Caputo derivative of lower bound $t_0$ and a forcing term $\mathcal F x_0$, which is given in Definition \ref{def_Forcing} below. The properties of the forcing term $\mathcal F x_0$ along with classical results on Caputo fractional differential equations allows to establish existence and uniqueness of solutions of problem \eqref{eq_nonlinearIVP}. We formalize this in Proposition \ref{prop_FromAstToC}.
	\begin{definition}\label{def_Forcing}
		 The functional $\mathcal{F}x_0:[t_0,\infty) \to \realnumbers^n$ of an initial function $x_0:(-\infty,t_0]\to \realnumbers^n$ is defined as
		\begin{align}\label{eq_ForcingIVP}
			\mathcal{F}x_0(t) := \int_{-\infty}^{t_0}\frac{(t-\tau)^{-\alpha}}{\Gamma(1-\alpha)}x_0'(\tau)\dsomething{\tau},
		\end{align} to which we will refer as the forcing term of the initial condition.
	\end{definition}
	
	The forcing term $\mathcal F x_0(t)$ of the initial condition is continuous and has the following decay property, which we will use in the analysis of a linear time-invariant version of \eqref{eq_nonlinearFODE} in Section \ref{sec_LTI}. The decay bound may be regarded as a fading-memory property, where the differential equation``forgets" events of the distant past. 
	\begin{lemma}\label{lemma_forcing}
		For an initial condition $x_0\in X(-\infty,t_0]$, {its forcing term $\mathcal F x_0$ is continuous on $[t_0,\infty)$} and may be bounded according to
		\begin{align}\label{eq_forcingBoundInftyNormNew}
			||\mathcal F x_0(t)||\leq C(t-t_0+\eta)^{-\alpha},\quad t\geq t_0
		\end{align} where $C>0$ is given by
		\begin{align}
			C = \frac{1}{\Gamma(1-\alpha)}\left(2||x_0||_\infty + \frac{\eta}{1-\alpha}\sup_{t\in [t_0-\eta,t_0]}||x_0'(t)||\right).
		\end{align}
	\end{lemma}
{ 
The proof of Lemma \ref{lemma_forcing} is given in Appendix \ref{appendix}.}
As motivated in equation \eqref{eq_SplitDerivative1}, one may split the Liouville-Weyl derivative into Caputo derivative and forcing term. Thus, a { strong} solution of the Liouville-Weyl initial value problem \eqref{eq_nonlinearIVP} always satisfies a Caputo initial value problem
\begin{align}\label{eq_rewrittenIVPrecastNew}
	\CderivativeInput{\alpha}{t_0}x(t) &= f(t,x(t))-\mathcal Fx_0(t) =: f_0(t,x(t)),\quad t \geq t_0, \quad x(t_0) = x_0(t_0).
\end{align}{ Typically, the $\AstderivativeInput{\alpha}{0}$-type differential equations are considered instead of $\CderivativeInput{\alpha}{0}$-type formulations, since the latter shows problems on the space $C[0,h]$, but the two formulations coincide on $AC[0,h]$, see \cite[Section 2.3]{Webb2025}. We formalize this for our case in the Proposition \ref{prop_FromAstToC} below. Consider the following example of $\CderivativeInput{\alpha}{0}$-type problems on $C[0,h]$.
{
\begin{example}
    Let $x:[0,1]\to\realnumbers$ be the Cantor ternary function \cite[pp.\ 96ff.]{GO2003}.
    This function is continuous, monotonically increasing, differentiable almost everywhere,
    and satisfies $x'(t)=0$ a.~e. Thus, $\CderivativeInput{\alpha}{0}x(t)=0$ for $t\in [0,1]$ and $x$ is a strong solution of the initial value problem $\CderivativeInput{\alpha}{0}x(t)=0$, $x(0)=0$. By a strong solution we mean a continuous function satisfying the differential equation and initial condition. 
\end{example}
}
To exclude this odd behavior, we may instead study the $\AstderivativeInput{\alpha}{0}$-type initial value problem
\begin{align}\label{eq_rewrittenIVP_Ast}
	\AstderivativeInput{\alpha}{t_0}x(t) &= f(t,x(t))-\mathcal Fx_0(t) =: f_0(t,x(t)),\quad t \geq t_0, \quad x(t_0) = x_0(t_0).
\end{align}
Since $f$ is continuous in its first variable and locally Lipschitz in its second variable, the same properties hold for $f_0$, because Lemma \ref{lemma_forcing} guarantees that the forcing term $\mathcal F x_0$ is continuous. We thus have existence of \emph{strong} continuous solutions in the sense of Theorem 6.5 of \cite{Diethelm2010}. These solutions may or may not coincide with solutions of the $\CderivativeInput{\alpha}{0}$-type problem \eqref{eq_rewrittenIVPrecastNew}. Of course, since the two operators $\CderivativeInput{\alpha}{0}$ and $\AstderivativeInput{\alpha}{0}$ coincide almost everywhere in the case of absolutely continuous functions, we have the following statement.
\begin{proposition}\label{prop_FromAstToC}
    Consider the unique solution $x\in C[t_0,t_0+h)$ for some $h>0$ of the $\AstderivativeInput{\alpha}{0}$-type problem \eqref{eq_rewrittenIVP_Ast}. If $x$ is absolutely continuous on $[t_0,t_0+h)$, then $x$ satisfies the $\CderivativeInput{\alpha}{0}$-type problem \eqref{eq_rewrittenIVPrecastNew} and thus also the Liouville-Weyl problem \eqref{eq_nonlinearIVP} almost everywhere on $[t_0,t_0+h)$.
\end{proposition}
We remark that, in the case when $f$ is linear in its second argument and autonomous in its first, as the solutions are given by a basis of Mittag-Leffler functions $E_\alpha(\lambda (\cdot)^\alpha)$, the solutions are always absolutely continuous. However, as the following example shows, the assumptions on $f$ being continous in its first argument and Lipschitz in its second, may lead to the $\AstderivativeInput{\alpha}{0}$-type problem having solutions which are differentiable nowhere, and thus not absolutely continuous.
\begin{example}
Consider the Weierstrass function \cite{Samko1994}
	\[
	W(t)=\sum_{k=0}^{\infty}\frac{1}{2^k}\sin(2^k\pi t), \quad t\in [0,1],
	\]
	which is continuous but nowhere differentiable on $[0,1]$. This function $W$ is H\"older continuous with exponent $\beta$ for all $\beta \in (0,1)$ and $W(0)=0$. Following the results of \cite{Samko1994}, $W$ thus has a continuous Caputo fractional derivative $\AstderivativeInput{\alpha}{0} W$ of order $\alpha$ for all $\alpha \in (0,1)$. Then, the function $W$ can be seen as a solution of the initial-value problem
	\begin{align}\label{eq_WeierstrassDE}
		\AstderivativeInput{\alpha}{0} x(t)=f(t), \quad t\in [0,1], \quad x(0)=0
	\end{align}
	with the continuous function $f(t)={}^{\rm C}\!D_0^\alpha W(t)$. Therefore, the solution of the differential equation \eqref{eq_WeierstrassDE} is not absolutely continuous. Notice, that $\CderivativeInput{\alpha}{0}W = \RLintegralInput{1-\alpha}{0}\operatorname D^1 W$ does not exist, as $\operatorname{D}^1 W$ also does not exist. The initial value problem \eqref{eq_WeierstrassDE} can be viewed as a situation in which the solutions of the  $\AstderivativeInput{\alpha}{0}$-type problem \eqref{eq_rewrittenIVP_Ast} and $\CderivativeInput{\alpha}{0}$-type problem \eqref{eq_rewrittenIVPrecastNew} do not coincide.
\end{example}
}

To give an impression of the proposed space of initial functions $X(-\infty,t_0]$, the following examples are provided.

\begin{example}\label{ex_initialFunctions} Consider the following functions $x_0:(-\infty,0]\to \realnumbers$, which may or may not be in $X(-\infty,0]$.
	\begin{enumerate}[label=(\roman*)]
		\item \label{item_exp-t}The function $x_0(t) = \exp(-t)$ does not belong to $X(-\infty,0]$,  as it is not bounded on $(-\infty,0]$. Its forcing term does not exist for any $t\geq 0$. Indeed, it holds that 
		\begin{align}
			\mathcal F x_0(t) &= -\int_{-\infty}^0\frac{(t-\tau)^{-\alpha}}{\Gamma(1-\alpha)}\exp(-\tau)\dsomething{\tau}= -\frac{\exp(-t)}{\Gamma(1-\alpha)}\int_{t}^\infty\frac{\exp(s)}{s^{\alpha}}\dsomething{s},
		\end{align} where the latter integral, in which the substitution $s = t -\tau$ has been employed, is unbounded. In fact, one may show that for existence of $\mathcal F x_0$, it is neccessary that $x_0(t) = O(|t|^{\alpha -\delta})$ for $t\to -\infty$ and some $\delta>0$.
		\item \label{item_exp+t} In contrast, it holds that $x_0 \in X(-\infty,0]$ for $x_0(t)= \exp(t)$, where its forcing term can be shown to read $\mathcal F x_0(t) = \exp (t) \Gamma(1-\alpha,t)/\Gamma(1-\alpha)$, where $\Gamma(\cdot,\cdot)$ is the upper incomplete gamma function.
		\item \label{item_mittag_leffler}$x_0 \in X(-\infty,0]$ for $x_0$ given by \begin{align}
			x_0(t)=\begin{cases}
				1, & t<-1\\
				E_\alpha(-(t+1)^\alpha),&-1\leq t<0
			\end{cases}
		\end{align}
		The function $x_0$ is continuous and bounded everywhere and its derivative exists everywhere except at $t=-1$ and is bounded in the neighborhood of $0$. {As it is continuously differentiable everywhere but at $t=-1$, with the singularity at $t=-1$ being integrable, it is also absolutely continuous.}
		\item \label{item_unbouded_der_at0}$x_0 \notin X(-\infty,0]$ for $x_0$: \begin{align}
			x_0(t) = \begin{cases}
				|t|^\alpha,&t\in (-1,0]\\
				1,&t\leq -1
			\end{cases},
			\quad x_0'(t) = \begin{cases}
				-\alpha|t|^{\alpha-1},&t\in (-1,0)\\
				0,&t\leq -1
			\end{cases}
		\end{align}
		Notice that $x_0'(t)$ grows unboundedly as $t\to 0-$. The forcing $\mathcal F x_0$ at time $t=0$ is given by
		\begin{align}
			\mathcal F x_0(0)&=\int_{-\infty}^0\frac{(-\tau)^{-\alpha}}{\Gamma(1-\alpha)}x_0'(\tau)\dsomething{\tau}\\
			&= \frac{-\alpha}{\Gamma(1-\alpha)}\int_{-1}^0{(-\tau)^{-\alpha}}{}(-\tau)^{\alpha-1}\dsomething{\tau}= \frac{-\alpha}{\Gamma(1-\alpha)}\int_{0}^1s^{-1}\dsomething{s},
		\end{align} where the latter integral is unbounded and $\mathcal F x_0(t)$ is thus is not defined at $t=0$.
        \item \label{item_constant}$x_0 \in X(-\infty,0]$ for $x_0$ given by \begin{align}
			x_0(t)=\begin{cases}
				1, & t<-1\\
				-x,&-1\leq t<0
			\end{cases},
			\quad x_0'(t) = \begin{cases}
				0, & t<-1\\
				-1,&-1\leq t<0
			\end{cases}
		\end{align} and one may explicitly find $\mathcal Fx_0(t) = \Gamma(2-\alpha)^{-1}(t^{1-\alpha}-(t+1)^{1-\alpha})$.
		\item \label{item_sine} $x_0 \in X(-\infty,0]$, for $x_0(t)= \sin(t)$.
	\end{enumerate}
\end{example}

The example \ref{item_exp-t} illustrates that exponentially decaying initial functions are not feasible, as these grow unboundedly in negative time. We will also observe in the next section that decaying solutions of linear versions of \eqref{eq_nonlinearFODE} possess an algebraic decay instead of an exponential decay. Exponential growth in case \ref{item_exp+t} is indeed possible and will be observed in the two following sections. In example \ref{item_mittag_leffler}, the initialization function $x_0$ is a typical trajectory generated by a linear fractional-order Caputo differential equation with finite lower bound. Although its first derivative is undefined at time $t=-1$, its forcing term is well defined. Example \ref{item_unbouded_der_at0} is an example showing that a singularity of $x_0'$ at $t_0$ can lead to nonexistence of $\mathcal F x_0$. This particular example of $x_0$ is bounded and uniformly continuous, showing that the space of initial functions of \eqref{eq_nonlinearFODE} needs to be chosen differently than in \cite{Hinze2020}. Example \ref{item_constant} allows for an explicit computation of the forcing term. As this work is motivated by the study of periodic solutions, example \ref{item_sine} is a periodic initial function. In fact, the Liouville-Weyl differential equation \eqref{eq_nonlinearFODE} can possibly posess a periodic solution, in contrast to fractional differential equations with finite lower bound \cite{Tavazoei2009}. Indeed, a fractional-order Caputo derivative with finite lower bound of a periodic function is not periodic \cite{Area2016}. A minimal example of a system with a periodic solution is given in the following. 

\begin{example}
	Consider the linear system $\CderivativeInput{\alpha}{-\infty}x(t)=Ax(t)$ for $t\in \realnumbers$ where
	\begin{align}
		A=\begin{bmatrix}
			\cos(\alpha \pi/2) & \sin(\alpha \pi /2) \\
			-\sin(\alpha \pi /2) & \cos(\alpha \pi /2) 
		\end{bmatrix} \in \realnumbers^{2\times 2}
	\end{align} with eigenvalues $\operatorname{eig} A = \exp(\pm i\alpha \pi /2)$ and $\alpha \in (0,1)$. The system admits the periodic solution 
	\begin{align}
		x(t) = \begin{bmatrix}
			\sin(t)- \frac{\cos(t)}{2\tan(\alpha\pi/2)}\\
			\frac{\sin(t)}{2\tan(\alpha\pi/2)}+\cos(t)
		\end{bmatrix}.
	\end{align} 
\end{example}

We intend to study stability properties of linear cases of the system \eqref{eq_nonlinearFODE} in the subsequent sections. We give the stability definition here, which is equivalent to the one of \cite{Hinze2020} and the classical literature on functional differential equations \cite{Hale1977}.
\begin{definition}[Stability of equilibria of autonomous systems]\label{def_stability}
	Consider the trivial equilibrium $x^\ast = 0$ of the system \eqref{eq_nonlinearFODE} in the autonomous case $f(t,x)=f(x)$ assuming w.l.o.g. $f(0) = 0$.
	 We call the trivial equilibrium $x^\ast = 0$ stable, if for all $\varepsilon>0$ there exists a $\delta >0$ such that the following holds:
	\begin{align}
		||x_0||_\infty <\delta\quad \Longrightarrow \quad ||x(t)||_2<\varepsilon \quad \forall t\geq 0,
	\end{align}
    {where $||r||_2$ denotes the 2-norm of a number $r\in \realnumbers^n$.}
\end{definition}
	\section{Linear Time-Invariant Fractional-Order Systems} \label{sec_LTI}
 In order to demonstrate the difficulties in finding analytical solutions of fractio\-nal-order Caputo differential equations with infinite memory, we first consider linear time-invariant differential equations. Recall the classical initial value problem with zero lower bound \cite{Diethelm2010}
\begin{align}\label{eq_CaputoDE_ZeroLowerBound}
{\AstderivativeInput{\alpha}{0}} z (t) = Az(t),\quad t\geq 0,\quad	z(0) = z_0 \in \realnumbers^n,
\end{align}
 where $\alpha \in (0,1)$ and $A\in \realnumbers^{n\times n}$.
It is well-known that for \eqref{eq_CaputoDE_ZeroLowerBound}, the Ansatz
\begin{align}
	\label{eq_ansatz_mittagLeffler} z(t) = E_\alpha(\lambda t^\alpha)v_z,\quad t\geq 0
\end{align}
yields (in the semisimple case) a complete basis of solutions of \eqref{eq_CaputoDE_ZeroLowerBound}, where $E_\alpha$ denotes the one parameter Mittag-Leffler function and $v_z\in \realnumbers^n$. Its Caputo derivative reads ${\AstderivativeInput{\alpha}{0}}z(t)=\lambda z(t)$ \cite{Diethelm2010}, yielding the characteristic equation
\begin{align}\label{eq_charEqn_MittagLeffler}
	\det (A-\lambda I) =0.
\end{align} Solutions $\lambda$ of the characteristic equations \eqref{eq_charEqn_MittagLeffler} are directly related to the stability properties of the trivial solution. This is stated in the following Theorem taken from \cite[Theorem 7.20]{Diethelm2010}  where the definitions of stability terms are given in the same reference.

\begin{theorem}
	Consider system \eqref{eq_CaputoDE_ZeroLowerBound} and let $\lambda_j$ be the eigenvalues of $A$ for $j = 1,\dots, n$. The trivial solution $z(t)= 0$ for all $t\geq 0$ is
	\begin{enumerate}[label=(\alph*)]
		\item asymptotically stable if and only if $|\arg \lambda_j|>\alpha \pi/2$ for all $j$.
		\item unstable if there exists $\lambda_j$ for which $|\arg\lambda_j|<\alpha \pi /2$.
	\end{enumerate}
\end{theorem}
Analogously, the initial value problem of a Liouville-Weyl differential equation (i.e. systems of form \eqref{eq_nonlinearIVP}) reads
\begin{subequations}\label{eq_CaputoDE_InfLowerBound}
	\begin{align}
		\CderivativeInput{\alpha}{-\infty} u (t) &= Au(t),\quad t\geq 0\\
		u(t) &= u_0(t),\quad t<0
	\end{align}
\end{subequations}
with $u_0 \in X(-\infty,0]$ and where both \eqref{eq_CaputoDE_ZeroLowerBound} and \eqref{eq_CaputoDE_InfLowerBound} share the same matrix $A\in \realnumbers^{n\times n}$.
We investigate forms of solutions of \eqref{eq_CaputoDE_InfLowerBound}, and their relation to solutions of  \eqref{eq_CaputoDE_ZeroLowerBound}. For \eqref{eq_CaputoDE_InfLowerBound} take the ansatz
\begin{align}\label{eq_ansatzExp}
	u(t) = \exp(s t) v_u,\quad t \in \realnumbers
\end{align}
where $v_u\in \realnumbers^n$ and $\mathrm{Re}(s) >0$. This is a natural choice, as its Liouville-Weyl derivative reads $\CderivativeInput{\alpha}{-\infty}u(t) = s^\alpha u(t)$ \cite{samko1993fractional}, yielding the characteristic equation
\begin{align} \label{eq_charEqn_exp}
	\det (A -s^\alpha I) = 0,
\end{align}
where $(\cdot)^\alpha$ denotes the principal value function
\begin{subequations}\label{eq_principalValue}
	\begin{align}
		p_\alpha:\quad &\mathbb C \to \mathcal M,\\
		& w \mapsto p_\alpha(w) = w^\alpha = |w|^\alpha \exp (i \alpha \arg w). \label{eq_principalValue_fcn}
	\end{align}
\end{subequations}
From \eqref{eq_principalValue_fcn}, we see that
\begin{align}
	\mathcal M = \{ m\in \mathbb C : \arg m \in (-\alpha \pi,+\alpha \pi]\}.
\end{align} Since systems \eqref{eq_CaputoDE_ZeroLowerBound} and \eqref{eq_CaputoDE_InfLowerBound} share the same matrix $A$, whenever \eqref{eq_charEqn_MittagLeffler} has a solution $\lambda \in \mathbb C$, also \eqref{eq_charEqn_exp} may have a solution $s$ given by
\begin{align}\label{eq_connectCharEquations}
	\lambda = s^\alpha.
\end{align}
	Relations \eqref{eq_principalValue} through \eqref{eq_connectCharEquations} imply the following result. Therein, the cases of $|\arg \lambda|$ are given in Figure \ref{fig_stability_region}.
	\begin{theorem} \label{thm_chareqn}
		Suppose \eqref{eq_charEqn_MittagLeffler} has a root $\lambda$ with
		\begin{enumerate}[label=(\alph*)]
			\item \label{case_unstable_soln} $|\arg \lambda| < \alpha \pi /2$: 
			Then \eqref{eq_charEqn_exp} has a solution $s$ with $\mathrm{Re}(s)>0$ and thus there exists a solution of \eqref{eq_CaputoDE_InfLowerBound} of the form \eqref{eq_ansatzExp} with $||u(t)||\to \infty,$  $t\to \infty$.
			\item \label{case_invalid_soln} $\alpha \pi /2<|\arg \lambda| \leq \alpha \pi$: 
			Then \eqref{eq_charEqn_exp} has a corresponding solution $s$ with $\mathrm{Re}(s)<0$, but the corresponding ansatz \eqref{eq_ansatzExp} is invalid and thus is no solution of \eqref{eq_CaputoDE_InfLowerBound}.
			\item \label{case_noSoln}$|\arg \lambda| > \alpha \pi$: 
			Then \eqref{eq_charEqn_exp} has no corresponding solution $s:~\lambda = s^\alpha$ and one cannot find a corresponding solution of the form \eqref{eq_ansatzExp} of \eqref{eq_CaputoDE_InfLowerBound}.
		\end{enumerate}
	\end{theorem}
    \begin{figure}[ht]
		\centering
		\includegraphics[page=1,width=0.8\textwidth]{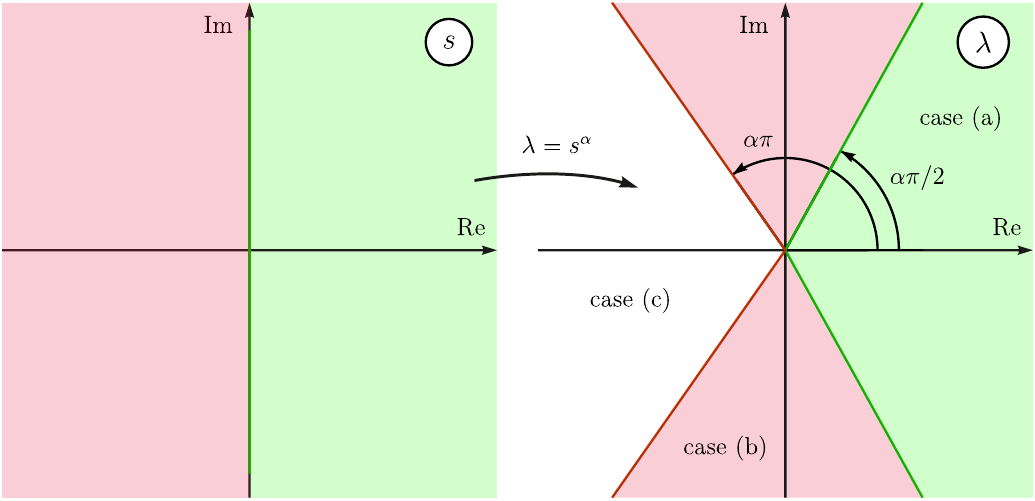}
		\caption{Mapping $p_\alpha (s) = s^\alpha = \lambda$ and cases of Theorem \ref{thm_chareqn}.}
				\label{fig_stability_region}
	\end{figure}
	
	Case \ref{case_unstable_soln} shows, that all diverging solutions of \eqref{eq_CaputoDE_ZeroLowerBound} have a corresponding diverging solution of \eqref{eq_CaputoDE_InfLowerBound} of exponential form. This is sufficient to show instability of the trivial solution in the sense of Definition \ref{def_stability}. In cases \ref{case_invalid_soln} and \ref{case_noSoln}, by definition of the exponential ansatz \eqref{eq_ansatzExp}, one cannot find converging solutions through the characteristic equation \eqref{eq_charEqn_exp}, as these are not in $X(-\infty,0]$. In the remainder of the section, we study the case, where the exponential ansatz fails.
    
		Solutions of \eqref{eq_CaputoDE_InfLowerBound} may be constructed irrespective of the solutions of the characteristic equation \eqref{eq_charEqn_exp}. As discussed in {Proposition \ref{prop_FromAstToC}}, each initial value problem \eqref{eq_CaputoDE_InfLowerBound} of Liouville-Weyl type may be rewritten as an inhomogenous initial value problem of {$\AstderivativeInput{\alpha}{0}$-type}
			\begin{align}\label{eq_IVPrewritten}
				{\AstderivativeInput{\alpha}{0}}u(t) = Au(t) - \mathcal{F}u_0 (t),\quad t\geq 0, \quad u(0) = u_0(0)
			\end{align}
		with the forcing term $\mathcal F u_0$ of Definition \ref{def_Forcing}.
		Applying the variation-of-constants formula \cite{Diethelm2010}, problem \eqref{eq_IVPrewritten} has the solution
		\begin{align}\label{eq_IVPsolution}
			u(t) = E_\alpha (A t^\alpha)u_0(0) -\int_0^t (t-\tau)^{\alpha-1} E_{\alpha,\alpha}(A(t-\tau)^\alpha) \mathcal F u_0(\tau) \dsomething{\tau},
		\end{align}
	where $E_{\alpha,\alpha}$ is the two-parameter Mittag-Leffler function. 
	{The general form of solutions \eqref{eq_IVPsolution} gives insight into solutions of \eqref{eq_CaputoDE_InfLowerBound}, even in cases where the exponential ansatz \eqref{eq_ansatzExp} fails (\ref{case_invalid_soln} and \ref{case_noSoln} above).  We provide an analysis of the scalar case $A<0$ in the following.}
	
	{
		{ \begin{theorem}\label{thm_Forcing}
			Consider the initial value problem \eqref{eq_IVPrewritten} along with its solution \eqref{eq_IVPsolution} for the scalar, real case $A <0$. Assume that $\mathcal F u_0$ of \eqref{eq_ForcingIVP} is continuous on $[0,\infty)$. We have:
			\begin{enumerate}
				\item If $\mathcal F u_0$ is bounded on $[0,\infty)$, then $u$ is bounded on $[0,\infty)$.
				\item If $\lim_{t\to\infty}\mathcal F u_0(t) = 0$, then $\lim_{t\to\infty} u(t) = 0$.
				\item If $\exists \eta \geq 0$ such that $\mathcal F u_0(t)= O(t^{-\eta})$ for $t\to\infty$ then $u(t)= O(t^{-\mu})$ for $t\to \infty$ where $\mu = \min \{\alpha,\eta\}.$
			\end{enumerate}
		\end{theorem}
	\begin{proof}
		Given the problem \eqref{eq_IVPrewritten}, we construct the two-dimensional initial value problem
		\begin{align}
			\begin{pmatrix}
				\CderivativeInput{\alpha}{0}u_1(t)\\
				\CderivativeInput{\alpha}{0}u_2(t)
			\end{pmatrix}= \tilde A \begin{pmatrix}
			u_1(t)\\
			u_2(t)
			\end{pmatrix}- \begin{pmatrix}
			\mathcal{F}u_0(t)\\
			\mathcal{F}u_0(t)
			\end{pmatrix},\quad t\geq 0,\quad \begin{pmatrix}
				u_1(0)\\
				u_2(0)
			\end{pmatrix} =  \begin{pmatrix}
			u_0(0)\\
			u_0(0)
			\end{pmatrix}
		\end{align}
		where $\tilde A = \operatorname{diag}(A,A)$. Clearly, the two components of this system are decoupled, and each of them is equivalent to the original initial value problem, so this formulation might appear unnecessarily complicated. It has the advantage, however, that we can now directly use the methods and results of \cite{Diethelm22}, which hold for planar systems.
		
		Specifically, we write down the characteristic function of the system \eqref{eq_IVPrewritten}, it being
		\begin{align}
			Q(s) = \det \left(\begin{pmatrix}
				s^\alpha & 0\\
				0 & s^\alpha
			\end{pmatrix}-\tilde A \right) = (s^\alpha - A)^2.
		\end{align}
		It follows that $Q(s) = 0$ for some $s \in \mathbb C$ if and only if $s^\alpha = A$. But now, for any $s\in\mathbb C$ we have that $\arg(s) \in (-\pi, \pi]$, and hence $\arg(s^\alpha) = \alpha \cdot \arg(s) \in (-\alpha\pi, \alpha\pi]$. 
		Since in our case $\alpha \in (0, 1)$, we thus conclude that $|\arg(s^\alpha)|\neq \pi$, and so the equation $s^\alpha = A$ cannot have a solution. Thus, the characteristic function does not possess any zeros in $\mathbb C$, and so in particular it does not have any zeros in the closed right half of the complex plane. This observation allows us to invoke \cite[Theorem 4]{Diethelm22} which immediately yields the desired result.
		\end{proof}}
		The above Theorem gives insight on solutions of \eqref{eq_CaputoDE_InfLowerBound} in the case where the exponential ansatz \eqref{eq_ansatzExp} fails, since \eqref{eq_charEqn_exp} has no solution. In fact, using Lemma \ref{lemma_forcing}, we can make the following statement.
		\begin{corollary}\label{corollary_scalarLTI_FODE}
			Consider the initial value problem \eqref{eq_CaputoDE_InfLowerBound} in the scalar, real case $A<0$. Its solution $u$ satisfies all properties listed in Theorem \ref{thm_Forcing}. In particular, it holds that  $u(t)= O(t^{-\alpha})$ for $t\to \infty$.
		\end{corollary}
		\begin{proof}
			By Lemma \ref{lemma_forcing}, the forcing term of the initial condition $u_0 \in X(-\infty,0]$ has all properties required by Theorem \ref{thm_Forcing}.
		\end{proof}
		Additionally to boundedness and decay, we next show stability of the present case. Therein, we use the following result.
	\begin{lemma}\label{lemma_forcingBoundC0}
		Assuming $||u_0||_\infty = \max_{t\leq 0}|u_0(t)|<\infty$ it holds for all $t\geq 0$ that
		\begin{align}
			|\mathcal{F}u_0(t)| \leq \frac{2||u_0||_\infty }{\Gamma(1-\alpha)}t^{-\alpha} .\label{eq_forcingBoundInftyNormC0}
		\end{align}
	\end{lemma}
	\begin{proof}
		The proof works the same as for Lemma \ref{lemma_forcing}, but without initially splitting the integral. 
	\end{proof}
	\begin{proposition}
	Consider the initial value problem \eqref{eq_IVPrewritten} along with its solution \eqref{eq_IVPsolution} for the scalar, real case $A <0$. The trivial equilibrium $u^\ast = 0$ is stable.
	\end{proposition}
\begin{proof}
	It holds that $|u(t)|\leq 3 ||u_0||_\infty$ for all $t\geq 0$, which we show in the following.
	As the behavior of the Mittag-Leffler function is well-studied, it remains to bound the influence of the forcing term
	\begin{align}
		\tilde {\mathcal F }u_0 (t) = -\int_0^t (t-\tau)^{\alpha-1} E_{\alpha,\alpha}(A(t-\tau)^\alpha) \mathcal F u_0(\tau) \dsomething{\tau}.
	\end{align}
	Applying the result of Lemma \ref{lemma_forcingBoundC0}, it holds that
	\begin{align}
		\left| \tilde {\mathcal F }u_0 (t)\right|&\leq \frac{2 ||u_0||_\infty}{\Gamma(1-\alpha)}\int_0^t (t-\tau)^{\alpha-1} |E_{\alpha,\alpha}(A(t-\tau)^\alpha)| \tau^{-\alpha} \dsomething{\tau}\\
		&\leq \frac{2 ||u_0||_\infty}{\Gamma(\alpha)\Gamma(1-\alpha)}\int_0^t (t-\tau)^{\alpha-1}  \tau^{-\alpha} \dsomething{\tau} \label{eq_boundForcingpart1}
	\end{align}
	Here we use that, since $E_{\alpha,\alpha}(A(\cdot)^\alpha)$ is monotonically decreasing whenever $A<0$ 
	(this immediately follows from the complete monotonicity of the function $E_{\alpha,\alpha}(-(\cdot))$
	shown in \cite[\S~4.9.2]{GorenfloEtAl2020}), it holds that
	\begin{align}
		\max_{\tau \in [0,t]}|E_{\alpha,\alpha}(A(t-\tau)^\alpha)| = E_{\alpha,\alpha}(0) = \frac{1}{\Gamma(\alpha)}.
	\end{align}
	Introducing the substitution $s = \tau /t$, the integral may be rewritten using basic properties of the Gamma function
	\begin{align}
	\int_0^t (t-\tau)^{\alpha-1}  \tau^{-\alpha} ~\dsomething{\tau} &= \int_0^1 (t-s t)^{\alpha-1}  (st)^{-\alpha} t~\dsomething{s}\\
		&= \int_0^1 (1-s)^{\alpha-1}  s^{-\alpha} ~\dsomething{s} =\Gamma(\alpha)\Gamma(1-\alpha).
	\end{align}
	Thus, the inequality \eqref{eq_boundForcingpart1} becomes
	\begin{align}
		\left| \tilde {\mathcal F }u_0 (t)\right|\leq 2||u_0||_\infty.
	\end{align} Finally, the solution may be upper bounded
	\begin{align} \label{eq_boundStability}
		|u(t)|\leq \max_{t\geq 0}E_\alpha(At^\alpha) |u_0(0)| + 2||u_0||_\infty \leq 3||u_0||_\infty
	\end{align}
	as the maximum is given by $E_\alpha(0) = 1$.
	
	We now use the bound \eqref{eq_boundStability} to show stability per Definition \ref{def_stability}. For any $\varepsilon>0$, take $\delta = \varepsilon/4$. With $|u(t)|\leq 3 ||u_0||_\infty$ we arrive at
	\begin{align}
		|u(t)|\leq 3 \delta =3/4 \varepsilon<\varepsilon.
	\end{align}
\end{proof}

	\section{Fractional-Order Linear Time-Periodic Systems}\label{sec_LTP}
	{
		We now assume that \eqref{eq_nonlinearIVP} either is time-autonomous $f(t,x) = f(x)$ or time-periodic, i.e. there exists a period time $\hat T>0$ s.t. $f(t,x) = f(t+\hat T,x)$ holds for all $x$ and $t$. Further, assume that \eqref{eq_nonlinearIVP} possesses a periodic solution $x_\mathrm{p}:\realnumbers\to \realnumbers^n$ where there exists a period time $T=m\hat T>0$ for some positive integer $m$ such that $x_\mathrm{p}(t+T)=x_\mathrm p(t)$ holds for all $t$. The LTP-FODE emerges when linearizing about this periodic solution $x_\mathrm p $ of \eqref{eq_nonlinearIVP}
		\begin{subequations} \label{eq_FracHillIVP}
			\begin{align}\label{eq_FracHillDiffEqn}
				\CderivativeInput{\alpha}{-\infty} y(t) &= J(t) y(t), \quad t \geq t_0\\
				y(t) &= y_0(t),\quad \quad \,\,t< t_0
			\end{align}
			where $x(t)-x_\mathrm p(t) \eqqcolon y(t)\in \realnumbers^n$ and $J(t)=J(t+T) = \frac{\partial f}{\partial x}(t,{x_\mathrm p(t)})\in \realnumbers^{n\times n}$ is $T$-periodic
		\end{subequations}
		and $y_0 \in {X(-\infty,t_0]}$. The coordinate  $y(t)$ is introduced to describe the local behavior of a small perturbation from the periodic solution $x_\mathrm p$. Thus, asymptotic stability and instability of the periodic solution is given whenever the equilibrium $y\equiv 0$ of \eqref{eq_FracHillIVP} is asymptotically stable and unstable, respectively.
		
		The choice of the lower bound $-\infty$ in the differential equation \eqref{eq_nonlinearFODE} is necessary when we want to study periodic solutions $x_\mathrm p(t) = x_\mathrm p(t+T)$ for all $t$ and some $T>0$ of \eqref{eq_nonlinearIVP}. A version of \eqref{eq_nonlinearIVP} with zero lower bound cannot have a periodic solution \cite{Tavazoei2009}. 
		
	} 
	We intend to study the behavior of the solutions of \eqref{eq_FracHillIVP}. To the best of our knowledge, there are no results on the form of solutions for this problem, i.e. no Floquet theory. Here, we will attempt to fill that gap.
	
	{
		\begin{theorem}\label{thm_fracHill}
			{ Let $J:[0,T]\to \realnumbers^{n\times n}$ and $p:[0,T]\to \realnumbers^n$ be $T$-periodic with the Fourier series expansions
				\begin{align}\label{eq_FracHillFourierForm}
					J(t) = \sum_{k=-\infty}^{+\infty}J_k\exp(i\omega k t),\quad p(t) = \sum_{k=-\infty}^{+\infty}p_k\exp(i\omega k t)
				\end{align}
				where $\omega = 2\pi / T$ and $J_k$, $p_k$ are complex quantities of appropriate dimension. Assume the following properties of $J$ and $p$.
				\begin{itemize}
					\item[\assJ] The Fourier series expansion of $J$ is convergent almost everywhere on $[0,T]$. 
					\item[\assp] The Fourier series expansion of $p$ is absolutely convergent, and $p \in C^1[0,T]$ with its derivative integrable.
				\end{itemize} A solution of \eqref{eq_FracHillIVP} that takes the form
				\begin{align}\label{eq_solutionFloquetForm}
					y(t) = \exp(\lambda t)p(t)
				\end{align}
				for all $t\in \realnumbers$ and for $\operatorname{Re}\lambda\geq 0$} fulfills the nonlinear eigenvalue problem
			\begin{align}\label{eq_fracHilleqn}
				0 = \hMat_\infty^{(\alpha)} (\lambda)\pVec_\infty,
			\end{align}
			where
			\begin{align}\label{eq_fracHillmat}
				\hMat_\infty^{(\alpha)}(\lambda) &= \begin{pmatrix}
					\ddots & \vdots & \vdots & \vdots & \\
					\dots & J_0 - (\lambda -i\omega)^\alpha I & J_{-1} & J_{-2} & \dots\\
					\dots & J_1 & J_0 -\lambda^\alpha I & J_{-1} & \dots\\
					\dots & J_{2} & J_{1} & J_0 - (\lambda +i\omega)^\alpha I & \dots\\
					& \vdots & \vdots & \vdots & \ddots
				\end{pmatrix},
				\\
				\pVec_\infty^\top &= \begin{pmatrix}
					\dots &	p_{-1}^\top & p_0^\top &	p_{+1}^\top &\dots
				\end{pmatrix}.
			\end{align}
			Herein, $(\cdot)^\alpha$ denotes the principal value and $I$ is the identity matrix of appropriate dimension.
			We call $\hMat_\infty^{(\alpha)}(\lambda)$ the fractional Hill matrix. 
		\end{theorem}
}
	\begin{proof}
		Recall, that the Liouville-Weyl derivative of the exponential function reads 
		\begin{align}
			\CderivativeInput{\alpha}{-\infty}\{\exp((\lambda + i\omega)t)\} = (\lambda + i\omega)^\alpha \exp((\lambda + i \omega )t),\quad \text{for all }t
		\end{align}
		whenever $\operatorname{Re}\lambda \geq 0$ and where $(\cdot)^\alpha$ denotes the principal value, $\alpha \in (0,1)$ \cite{samko1993fractional}. One may compute the fractional derivative of the Floquet form \eqref{eq_solutionFloquetForm} by interchanging the order of summation and fractional differentiation, justified by Assumption~\assp,
		\begin{align}
			\CderivativeInput{\alpha}{-\infty}y(t) &= \CderivativeInput{\alpha}{-\infty}\left\{\sum_{k=-\infty}^{+\infty}p_k \exp((\lambda + ik\omega)t)\right\}\\
			&= \sum_{k=-\infty}^{+\infty}p_k \CderivativeInput{\alpha}{-\infty}\left\{\exp((\lambda + ik\omega)t)\right\}\\
			&= \sum_{k=-\infty}^{+\infty}p_k (\lambda + ik\omega)^\alpha \exp((\lambda + ik\omega)t). \label{eq_RHSmanip}
		\end{align}
		Next, rewrite the right-hand side of \eqref{eq_FracHillDiffEqn} using \eqref{eq_FracHillFourierForm} and \eqref{eq_solutionFloquetForm}
		\begin{align}
			J(t)y(t) &= \left(\sum_{j=-\infty}^{+\infty}J_j\exp(i\omega j t)\right)\left(\sum_{k=-\infty}^{+\infty}p_k\exp((\lambda + ik\omega ) t)\right)\\
			&= \sum_{k=-\infty}^{+\infty}\sum_{j=-\infty}^{+\infty}J_jp_k\exp((\lambda + i(k+j)\omega ) t)\\
			&= \sum_{k=-\infty}^{+\infty}\left(\sum_{j=-\infty}^{+\infty}J_jp_{k-j}\right)\exp((\lambda + ik\omega ) t)\label{eq_LHSmanip}
		\end{align}
		where in the last line we performed the index shift $k\mapsto k-j$, justified by Assumptions \assJ~and \assp. Substituting \eqref{eq_RHSmanip} and \eqref{eq_LHSmanip} into \eqref{eq_FracHillDiffEqn}, one obtains
		\begin{align}
			0=\sum_{k=-\infty}^{+\infty}\left(p_k(\lambda + ik\omega)^\alpha-\sum_{j=-\infty}^{+\infty}J_jp_{k-j}\right)\exp((\lambda + ik\omega ) t).
		\end{align}
		The sum over index $k$ vanishes for all $t$ iff every summand of the outer sum vanishes. With the above definition of $\pVec_\infty$, the $k$-th term reads
		\begin{align}\label{eq_kthTermProof}
			0 = (\lambda + ik\omega)^\alpha \begin{bmatrix}
				\cdots &0&I&0&\cdots
			\end{bmatrix}\pVec_\infty-\begin{bmatrix}
				\cdots &J_{1}&J_0&J_{-1}&\cdots
			\end{bmatrix}\pVec_\infty
		\end{align}
		where the middle entry of the two wide matrices are paired with the $k$-th entry of $\pVec_\infty$. Rewriting \eqref{eq_kthTermProof} into matrix notation for all $k$, one obtains \eqref{eq_fracHilleqn} and the proof is complete.
	\end{proof}
	The technical Assumptions \assJ~and \assp~on $J$ and $p$, respectively, ensure that the reformulation into the nonlinear eigenvalue problem \eqref{eq_fractionalHillProblem} is valid. Both assumptions are following the approach of \cite{Zhou2002} on first-order systems. For a given function $J$, Assumption \assJ~is easily checked. However, as the function $p$ is a priori unknown and is only approximated via truncated Fourier series in the following, Assumption \assp~remains a technical assumption.
	
	A nonlinear eigenvalue problem such as \eqref{eq_fractionalHillProblem} as a characteristic equation is not uncommon. We have seen in Section \ref{sec_LTI}, that a similar but simpler eigenvalue problem occurs for time-invariant systems. Another example of nonlinear eigenvalue problems as characteristic equations shows up in delay differential equations \cite{Stepan2011} for both the time invariant and time periodic cases. Therein, the nonlinearity in the eigenvalue problem is an exponential of the eigenvalue $\lambda$.
	{
			
		The constraint $\operatorname{Re} \lambda \geq 0$ is neccessary, since otherwise the Liouville-Weyl derivative does not exist. {This shows, that an extension of Floquet form solutions as in \cite{Traversa2020} to the system \eqref{eq_FracHillIVP} is generally not possible.} For a certain class of systems with memory terms, results on forms of solutions are available \cite{Traversa2020}. However, fractionally damped mechanical systems~\cite{Hinze2020}
		\begin{align}
			M(q)\ddot q - h\left(q,\CderivativeInput{\alpha}{-\infty}q,\dot q\right) = 0
		\end{align}
		just barely escape the requirements of \cite{Traversa2020}, since the memory kernel $K(t,\tau) = (t-\tau)^{-\alpha}/\Gamma(1-\alpha)$ does not satisfy $\int_{-\infty}^{t}|K(t,\tau)|\dsomething{\tau}<\infty$ for any $t$.
	\begin{remark}[Groups of Floquet exponents]\label{rem_GroupOfLambda}
		By the structure of the fractional Hill matrix \eqref{eq_fracHillmat}, given a solution~$\lambda$ which makes $\hMat_\infty^\alpha(\lambda)$ rank deficient, $\lambda + ik\omega$ for all $k \in \mathbb{Z}$ is also a solution. Hence, the Floquet exponents come in groups, just as in the integer-order case $\alpha = 1$ \cite{Bayer2023a}. This comes as no surprise, as the property is a consequence of the ansatz \eqref{eq_solutionFloquetForm}.
	\end{remark}}
	\begin{remark}[Floquet exponents with negative real part] \label{rmk_negativeFloquetExponents}
		There are cases, where the fractional Hill problem \eqref{eq_fracHilleqn} has eigenvalues $\lambda$ with negative real part. Consider a constant system matrix $J(t) = J_0 \in \realnumbers^{n\times n}$, which is $T$-periodic for any $T>0$. Choose $J_0$ such that it has a linear eigenvalue $\mu: \det(J_0-\mu I) = 0$ with $\alpha \pi /2<|\arg \mu| \leq \alpha \pi$. Then, as stated in Theorem \ref{thm_chareqn}, there exists a $\lambda:\mu = \lambda^\alpha$ with $\pi /2<|\arg \lambda| \leq \pi$ and thus $\operatorname{Re}\lambda <0$. This $\lambda$ is an eigenvalue of \eqref{eq_fracHilleqn}, since the block-wise center row of \eqref{eq_fracHillmat} reads
		\begin{align}
			&\left(\cdots \quad 0 \quad 0 \quad J_0 -\lambda^\alpha I\quad 0\quad 0\quad \cdots\right)\\
			= &\left(\cdots \quad 0 \quad 0 \quad J_0 -\mu I\quad 0\quad 0\quad \cdots\right)
		\end{align}
		and is thus rank-deficient.
	\end{remark}
{
	\begin{remark}
		{ One could consider the differential equation\begin{subequations}\label{eq_FracHillDiffEqn_ZeroLB_IVP}
				\begin{align}\label{eq_FracHillDiffEqn_ZeroLB}
					\AstderivativeInput{\alpha}{0} \tilde y(t) &= J(t) \tilde y(t), \quad t >0\\
					\tilde y(0) &= \tilde y_0
				\end{align}
			\end{subequations} where $J$ is $T$-periodic. Inspired by the form of solutions of LTI first-order DEs and LTI-FODEs with lower bound zero, taking an Ansatz as \eqref{eq_solutionFloquetForm} but with Mittag-Leffler functions, will greatly complicate setting up a Hill-type problem due to the lack of a multiplicative identity of the Mittag-Leffler function \cite{Diethelm2010}
			\begin{align}
				E_\alpha(a)E_\alpha(b)\neq E_\alpha(a+b).
		\end{align}
	However, it is clear that there exists a connection between the solutions of problems \eqref{eq_FracHillDiffEqn_ZeroLB_IVP} and \eqref{eq_FracHillIVP}. Indeed, when the Liouville-Weyl problem \eqref{eq_FracHillIVP} is reformulated to a Caputo initial value problem as done in \eqref{eq_rewrittenIVPrecastNew}, the only difference is the time-dependent forcing term $\mathcal Fy_0(t)$ and thus \eqref{eq_FracHillDiffEqn_ZeroLB_IVP} and \eqref{eq_FracHillIVP} share a homogenous solution.}
	\end{remark}}
	
{
		\begin{remark}[Truncated fractional Hill problem]
		For a numerical approximation, one can consider the truncated fractional Hill matrix of truncation order $N$
		\begin{align}\label{eq_FractionalHillMatTruncated}
			\hMat_{N}^{(\alpha)}(\lambda) &= \begin{pmatrix}
				J_0 - (\lambda -iN\omega)^\alpha I & \cdots & J_{-2N}\\
				\vdots & \ddots & \vdots \\
				J_{2N} & \cdots & J_0 - (\lambda +iN\omega)^\alpha I
			\end{pmatrix},\\
			\pVec_N^\top &= \begin{pmatrix}
				p_{-N}^\top &\dots &p_{N}^\top
			\end{pmatrix},
		\end{align}where $\hMat_{N}^{(\alpha)}(\lambda)\in \mathbb C^{n(2N+1)\times n(2N+1)}$ and $\pVec_N\in \mathbb C^{n(2N+1)}$. One may find a $\lambda$ with non-negative real part for which the matrix $\hMat_N^{(\alpha)}(\lambda)$ becomes singular, i.e.
		\begin{align}\label{eq_fractionalHillProblem}
			0 = \det \hMat_{N}^{(\alpha)}(\lambda).
		\end{align} This truncated fractional Hill problem then allows for computation of $\pVec_N$ as an element of the nullspace of $\hMat_{N}^{(\alpha)}(\lambda)$ and thus for construction of the solution \eqref{eq_solutionFloquetForm}. The equivialent formulation
		\begin{equation}
		\hMat_{N}^{(\alpha)}(\lambda) \pVec_N = 0
		\end{equation}
		belongs to the class of Nonlinear Eigenvalue Problems (NEPs)  described by \cite{Guettel2017}
		\begin{equation}
		M(z)v=0,
		\end{equation}
		where the eigenvalue $z$ may appear nonlinearly, but the corresponding eigenvector $v$ linearly in the equation. Numerical treatment of these problems is extensively discussed in \cite{Guettel2017}.
	\end{remark}}

{	
	\begin{remark}
		In the case $\alpha =1$, $\lambda$ is a linear eigenvalue of the Hill matrix. Indeed, $\hMat_N^{(1)}(\lambda) = \hMat_N -\lambda I$ where $\hMat_N = \hMat_N^{(1)}(0)$ is the classical Hill-matrix \cite{Bayer2023a}.
	\end{remark}}
Before providing example systems and solving their respective fractional Hill problem numerically, we collect some properties of the truncated fractional Hill problem. In particular, we state results on the location of the eigenvalues $\lambda$. Also, Remark \ref{rmk_negativeFloquetExponents} remains true for the truncated problem.
{
		\begin{proposition}\label{prop_ComplexConj}
			The solutions of the truncated fractional Hill problem \eqref{eq_fractionalHillProblem} come in complex conjugate pairs $(\lambda,\bar\lambda)$. 
		\end{proposition}
		\begin{proof}
			First note two things:
			\begin{enumerate}
				\item Since $J(t)$ is real, its Fourier coefficients \eqref{eq_FracHillFourierForm} satisfy $\overline{J_k} = J_{-k}$ for all $k$, where $\overline{(\cdot)}$ denotes the element-wise complex conjugate.
				\item The principal root function $\lambda^\alpha = |\lambda|^\alpha \exp(i\alpha \arg \lambda)$ commutes with complex conjugation, i.e. $\overline{(\lambda^\alpha)} = (\overline \lambda)^\alpha$. 
			\end{enumerate}
			For a quadratic matrix $M$, it holds that $\operatorname{rank}M = \operatorname{rank}\overline{M}$ and we may take $M = \hMat_{N}^\alpha(\overline\lambda)$
			\begin{align}
				\operatorname{rank} \hMat_{N}^\alpha(\overline\lambda) = \operatorname{rank}\overline{\hMat_{N}^\alpha(\overline\lambda)}
			\end{align}
			where the latter matrix -- using the two statements above -- reads
			\begin{align}\label{eq_proofComplexConjMat1}
				\overline{\hMat_{N}^\alpha(\overline\lambda)} = \begin{pmatrix}
					J_0 - (\lambda +iN\omega)^\alpha I & \cdots & J_{2N}\\
					\vdots & \ddots & \vdots \\
					J_{-2N} & \cdots & J_0 - (\lambda -iN\omega)^\alpha I
				\end{pmatrix}.
			\end{align} Swapping columns and rows blockwise, we see that matrix \eqref{eq_proofComplexConjMat1} has the same rank as $\hMat_{N}^{(\alpha)}(\lambda)$, completing the proof.
		\end{proof}
		\begin{proposition}[Gershgorin Regions]\label{prop_Gersh}
			The eigenvalues $\lambda$ of $\hMat_{N}^{(\alpha)}(\lambda)$ are located in the union of balls centered on the imaginary axis, with their centers spaced by $i\omega$ and radii dependent on the Fourier coefficients $J_k$. See Figure \ref{fig_Gershgorin_FracHill}.
		\end{proposition}
		\begin{proof}[Sketch of Proof]
			For the sake of simplicity, consider the scalar case $n=1$. We apply \cite[Thm 2.19]{Guettel2017} noting that the fractional Hill matrix
			\begin{align}
				\hMat_{N}^{(\alpha)}(\lambda) = D(\lambda) + R
			\end{align}
			can be split into a diagonal matrix $D(\lambda)$ and a remainder $R = const$. For the center row, we obtain
			\begin{align}
				|\lambda|^\alpha \leq \sum_{k\in \mathcal K_0} |J_k| \eqqcolon r_0\quad \Rightarrow\quad |\lambda|\leq r_0^{1/\alpha}
			\end{align}
			for some appropriate set $\mathcal K_0 \subset \{-2N,\dots ,+2N\}$, and find the first ball. Consider the upper neighbor of the center row, which yields the ball
			\begin{align}
				|\lambda-i\omega|\leq r_{-1}^{1/\alpha}
			\end{align}
			shifted from the origin and with $r_{-1} = \sum_{k\in \mathcal K_{-1}} |J_k|$ for an appropriate subset $\mathcal{K}_{-1} \subset \{-2N,\dots ,+2N\}$. The other balls follow analogously, producing a number of $2N+1$ balls.
		\end{proof}
		\begin{figure}[ht]
			\centering
			\includegraphics[width=6cm]{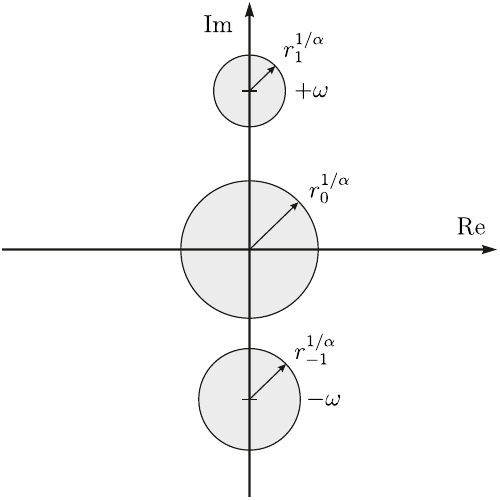}
			\caption{Location of eigenvalues $\lambda$ of $\hMat_{N}^{(\alpha)}(\lambda)$ in the complex plane.}
			\label{fig_Gershgorin_FracHill}
		\end{figure}}
		
		{
			\begin{example}[Scalar System]\label{ex_scalar}
			As an example of system \eqref{eq_FracHillDiffEqn}, consider $J(t) = a +b \sin (t)$ with real scalars $a$ and $b$. That is, we are interested in solutions of the IVP
			\begin{subequations}\label{eq_exampleSys}
					\begin{align}
							\CderivativeInput{\alpha}{-\infty}y(t) &= (a +b \sin (t))y(t),&& t>0\\
							y(s) &= y_0(s), && s\leq 0.
						\end{align}
				\end{subequations}
		\end{example}
		There is no analytical reference solution available and one has to resort to numerical solutions. One observes stable or unstable behavior of the origin $y\equiv 0$ in Figure \ref{fig_simulationExample}, depending on the choice of the parameters $a$ and $b$. Notice that the first-order case $\alpha =1$ of \eqref{eq_exampleSys} has an analytical solution obtained by variation of constants, whose asymptotic properties are solely dependent on parameter $a$.
        \begin{figure}[ht]
			\centering
			\begin{subfigure}[b]{0.45\textwidth}
				\centering
				\includegraphics[width=\textwidth]{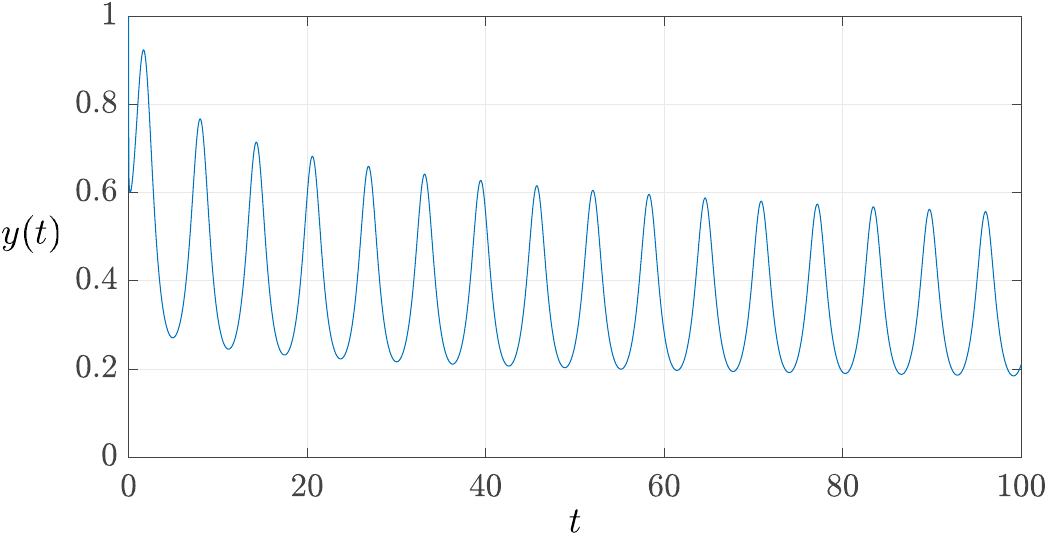}
				\caption{$b=1$.}
			\end{subfigure}
			\hfill
			\begin{subfigure}[b]{0.45\textwidth}
				\centering
				\includegraphics[width=\textwidth]{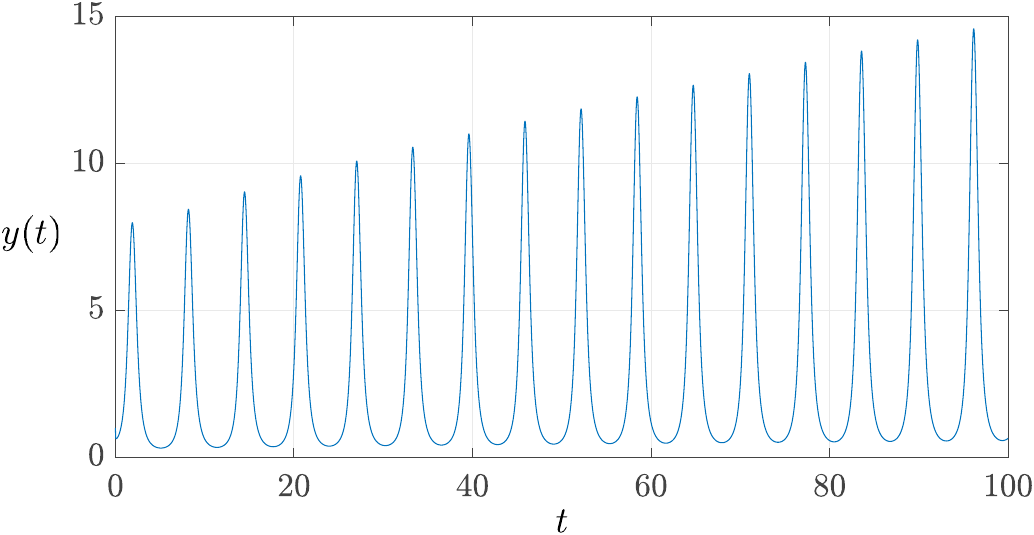}
				\caption{$b=2$.}
			\end{subfigure}
			\caption{Numerical solutions of Example \ref{ex_scalar} with constant initial condition $y_0(t) = 1$ for $t\leq 0$, $\alpha = 1/2$, $a =-1$ and varying parameter $b$. We used Garappa's implementation \texttt{fde12.m} in Matlab of the FracPECE scheme \cite{diethelm1998fracpece}.}
			\label{fig_simulationExample}
		\end{figure}
        }
		
		{ Having established a numerical reference solution, we now investigate the behavior of the fractional Hill matrix $\hMat_{N}^{(\alpha)}$. The Fourier coefficients of $J(t)$ of \eqref{eq_exampleSys} are given as
				\begin{align}
						J_0 = a,\quad J_1 = -J_{-1} = -ib/2,\quad J_k =0\quad \mathrm{for}~k\neq -1,0,1.
					\end{align} Figure \ref{fig_SimVsRes_Res} illustrates, that systems with a stable or unstable numerical reference solution do not or do admit zeros of $\det \hMat_{N}^{(\alpha)}(\lambda)$, respectively.}
		
		Figure \ref{fig_FloquetExp} shows the result of a grid zero search of $\det \hMat_{N}^{(\alpha)}$ for a fixed parameter set using Matlab's function \texttt{fzero}. The instabilty of the origin $y\equiv 0$ seen in reference numerical time-integration solutions of Figure \ref{fig_simulationExample} is confirmed by finding zeros $\lambda$ with positive real part. Moreover, the solutions $\lambda$ seem to be located on a line parallel to the imaginary axis, spaced roughly in the imaginary direction by the systems angular frequency $\omega =1$, confirming the statement of Remark \ref{rem_GroupOfLambda}.
		}
		
		{Having found an eigenvalue of $\hMat_{N}^\alpha$, one may also compute its corresponding eigenvector $\pVec_N$ and construct the trajectory $\HillSoln$ as \eqref{eq_solutionFloquetForm}. In order to confirm the computed trajectory, one may take it as an initial condition $y_0(s) = \HillSoln(s)$ for $s\leq 0$ of \eqref{eq_exampleSys} and numerically compute its trajectory $\simSoln$ using an integration scheme for the differential equation. The result is shown in Figure \ref{fig_TimeStepingVerification_HillVsNum}.}
		\begin{figure}[ht]
			\centering
			\begin{subfigure}[b]{0.45\textwidth}
				\centering
				\includegraphics[width=\textwidth]{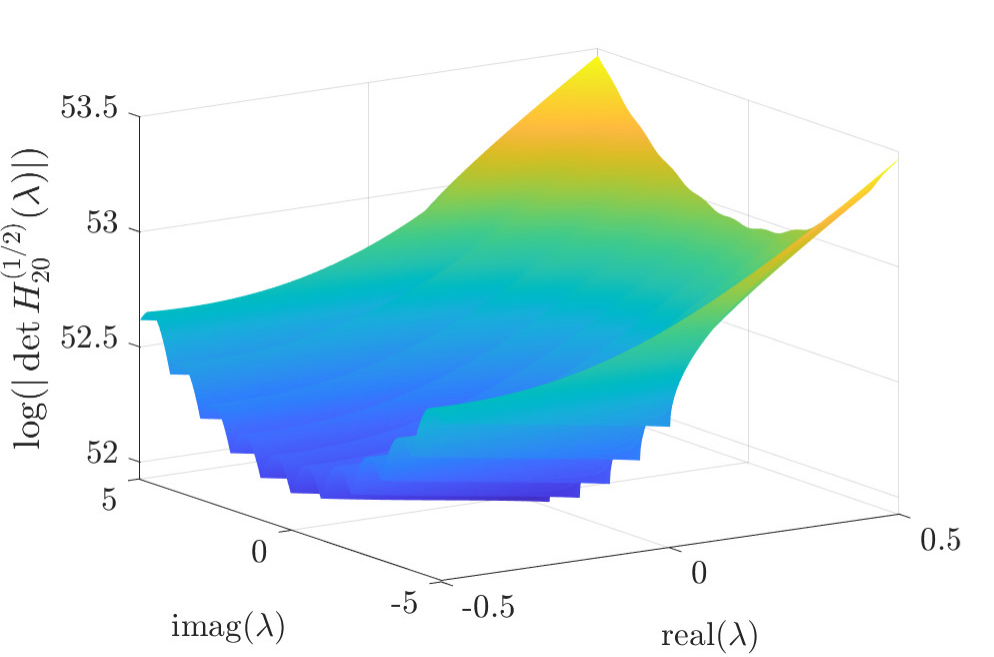}
				\caption{$b=1$.}
				\label{fig_SimVsRes_Res_b1}
			\end{subfigure}
			\hfill
			\begin{subfigure}[b]{0.45\textwidth}
				\centering
				\includegraphics[width=\textwidth]{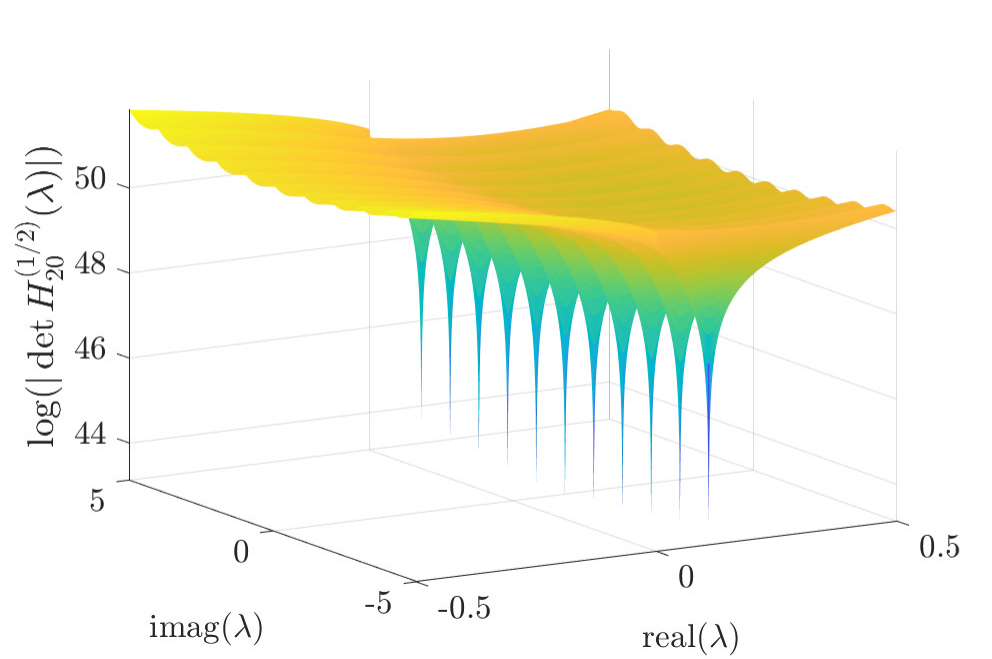}
				\caption{$b=2.5$.}
				\label{fig_SimVsRes_Res_b2p5}
			\end{subfigure}
			\caption{Logarithm of the determinant of the fractional Hill matrix $\hMat_N^{(\alpha)}$ of the system \eqref{eq_exampleSys} with $\alpha = 1/2$, $a =-1$, $N=20$ and varying parameter $b$. See Figure~\ref{fig_simulationExample} for the respective numerical time integration solutions.}
			\label{fig_SimVsRes_Res}
		\end{figure}
		\begin{figure}[ht]
			\centering
            \includegraphics[]{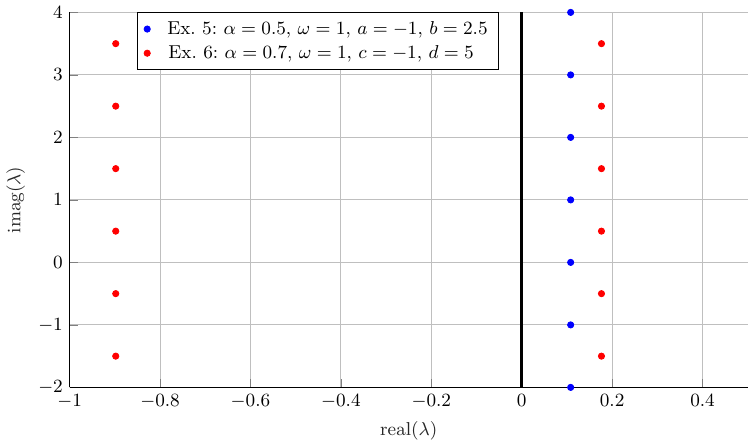}
			\caption{Numerically found solutions $\lambda$ of $\det \hMat_N^{(\alpha)}(\lambda) = 0$ of Examples \ref{ex_scalar} and \ref{ex_mathieu}.}
			\label{fig_FloquetExp}
		\end{figure}
		{
			\begin{figure}[ht]
				\centering
				\includegraphics[width = .6\textwidth]{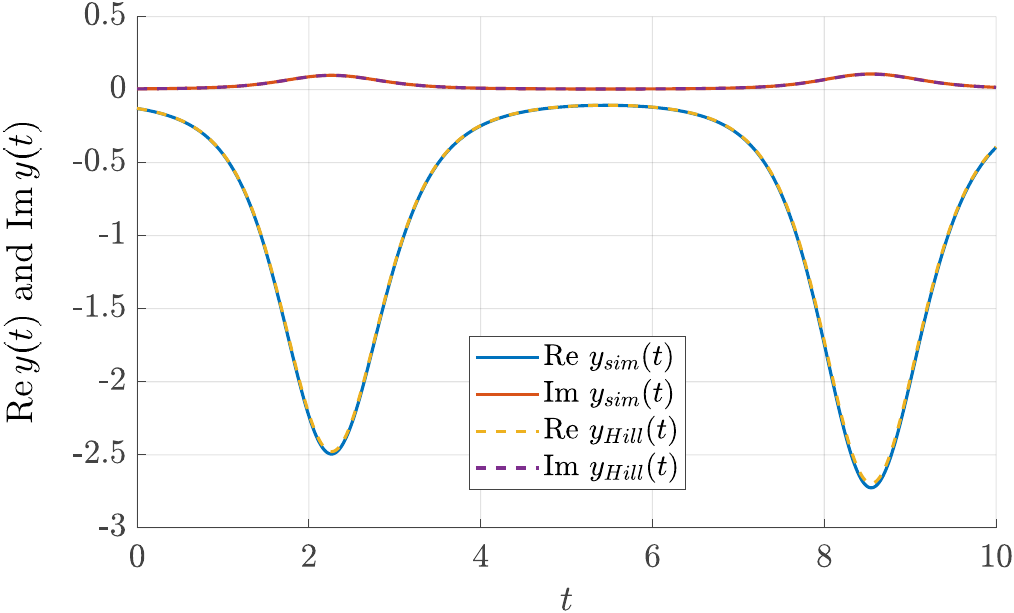}
				\caption{ Results of numerical time integration $\simSoln(t)$ of the system \eqref{eq_exampleSys} with $\alpha = 1/2$, $a =-1$, $N=10$ and $b=2.2$, compared with the solution obtained by the Hill method $\HillSoln(t)$. We used the numerical scheme of \cite{Chaudhary2024} implemented by Afshin Farhadi for problems of the form $\CderivativeInput{\alpha}{0}x = f(t,x)$ to initial condition $x(0) = x_0$. We implemented the infinite-memory initial condition $y_0$ as a forcing term on the right-hand side. }
				\label{fig_TimeStepingVerification_HillVsNum}
		\end{figure}}
		\begin{example}[Fractional Mathieu type system]\label{ex_mathieu}
			Consider the $n=2$-dimensional example system matrix of \eqref{eq_FracHillDiffEqn}
			\begin{align} \label{eq_FracMathieu}
				J(t) = \begin{bmatrix}
					0~~& 1\\
					c+d\sin(t)~~& 0
				\end{bmatrix}
			\end{align}
			with real parameters $c,d$ and Fourier coefficients
			\begin{align}
				J_0 = \begin{bmatrix}
					0~~ & 1\\
					c~~ & 0
				\end{bmatrix},\quad
				J_{-1} = \begin{bmatrix}
				0~~ & 0\\
				\frac{id}{2}~~ & 0
				\end{bmatrix}=\bar{J_1},\quad J_k = 0,\quad |k|\geq 2.
			\end{align}
			Results of numerical eigenvalue search of the truncated fractional Hill matrix are shown in Figure \ref{fig_FloquetExp}. We identify two groups of eigenvalues, one of which posesses a positive real part, indicating the existence of a diverging solution. For the group with negative real parts, one cannot find a respective converging solution. The existence of these eigenvalues confirms the statement of Remark \ref{rmk_negativeFloquetExponents}.
		\end{example}
		\section{Conclusion and Open Questions}\label{sec_outlook}
		{The properties of ansatz functions of both LTI and LTP fractional-order systems have been studied in the previous sections. In Table \ref{tab_summary}, we summarize the presented results on the two system types. For comparison, we also provide the case of first-order systems ($\alpha = 1$). As discussed in Section \ref{sec_LTI}, there is a clear relationship between the fractional-order LTI systems with zero and infinite lower bound $t_\mathrm L$, i.e. Captuto and Liouville-Weyl differential equations, respectively. It is possible to map all diverging solutions via an exponential ansatz, as stated in Theorem \ref{thm_chareqn}. For decaying solutions, the available results on Caputo systems help to better understand Liouville-Weyl systems. Indeed, using the fading memory property of Lemma \ref{lemma_forcing} as well as the variation of constants formula \eqref{eq_IVPsolution} allows for the analysis of the decaying solutions in Theorem \ref{thm_Forcing} and Corollary \ref{corollary_scalarLTI_FODE}. This analysis was performed in the scalar real negative case. Future research may try to extend this to the general complex case, where the exponential ansatz does not yield a conclusion in Theorem \ref{thm_chareqn}, i.e. cases \ref{case_invalid_soln} and \ref{case_noSoln}. 
			
		For the LTP case, a first-order differential equation classically yields the infinite-dimensional Hill eigenvalue problem. Caputo differential equations have no known form of solutions in the LTP case. This means that one cannot proceed as for LTI systems, that is, use knowledge on Caputo differential equations to investigate the Liouville-Weyl case. Consequently, the exponential ansatz and resulting nonlinear eigenvalue problem of Theorem \ref{thm_fracHill} yield a novel path to accessing solutions of LTP fractional differential equations. Therein, we observe the same limitation of the exponential ansatz as for LTI systems. One can only construct exponentially growing solutions, but not the algebraically decaying solutions. Besides the infinite-dimensional eigenvalue problem, we proved properties of the truncated eigenvalue problem in Propositions \ref{prop_ComplexConj} and \ref{prop_Gersh}. The proposed fractional Hill method enables stability analysis, whereas already established transition curve analysis is limited to finding periodic solutions, but gives no information on stability changes and bifurcations in parameter dependence.
		
		We close this paper with the following open questions: 
			\begin{enumerate}
				\item Is the absence of eigenvalues with non-negative real part, i.e. $\{\lambda: \mathrm{Re}\lambda \geq 0\}$, of the fractional Hill eigenvalue problem \eqref{eq_fracHilleqn} conclusive for stablity? Practice has shown, that stable systems either admit eigenvalues in the left half plane only or none at all. This suggests, that the fractional Hill method may be conclusive also for stability.
				\item How are the stability properties of the zero- and infinite-lower bound system related? Since the two system types only differ through the forcing term $\mathcal F x_0(t)$, which satisfies the algebraic decay property of Lemma \ref{lemma_forcing}, it may be possible to generalize the statements of Theorem \ref{thm_Forcing}. However, this does not extend to the most general case of nonlinear systems, as the forcing $\mathcal F x_0(t)$ may be arbitrarily large and potentially drive the solutions outside of regions of attraction, where the Caputo system would stay. 
				\item In the case of LTP-FODEs, how can one construct a monodromy operator and is it possible to find Floquet multipliers? Since the history of trajectories is a part of FODEs, we conjecture that a monodromy operator $\Phi_T:X(-\infty,0)\to X(-\infty,T)$ is required for the study of LTP systems analogous to the case of functional differential equations \cite{Hale1977}. The spectrum of this operator will map the stability properties of the respective system. Consequently, the study of fundamental solution operators for LTP-FODEs is a neccessary first step.
			\end{enumerate}}{ }
			\begin{table}[ht]
				\centering
				\caption{ Linear Differential Equations of the form $\CderivativeInput{\alpha}{t_\mathrm L} \zeta= A(t)\zeta$, where $A(t)=const$ and $A(t) = A(t+T)$ periodic yields a linear time-invariant (LTI) and linear time-periodic (LTP) type, respectively. Every Ansatz solution $\zeta(t)$ yields a corresponding eigenvalue problem (EVP), where $\rho$ may be restricted. }\vspace{.15cm}
					\begin{tabular}{r|ccc|ccc}\hline
					System& {LTI} &{LTI} &{LTI}  & {LTP}&{LTP} &{LTP}                            \\ \hline
					order $\alpha$	& {1} & {$(0,1)$} & {$(0,1)$}   & {1} & {$(0,1)$} & {$(0,1)$}    \\
					$t_\mathrm L$	& {-} & {$0$} & $-\infty$ & {-} & {$\quad 0\quad $} &  $-\infty$\\ 
					$\zeta(t)$& {$\exp(\rho t)v$} & {$E_\alpha(\rho t^\alpha)v$} &$\exp(\rho t)v$  & {$\exp(\rho t)p(t)$} & {?} & $\exp(\rho t)p(t)$ \\ 
					EVP& {$Av = \rho v$} & {$Av = \rho v$}                         &$Av=\rho^\alpha v$  & {$\hMat_\infty v = \rho v$} & {-} & $\hMat_\infty^{(\alpha)}(\rho)v = 0$ \\ 
					$\operatorname{dom}\rho$ 	& {$\rho \in \mathbb C$} &{$\rho \in \mathbb C$}                         & $\operatorname{Re} \rho \geq 0$ & {$\rho \in \mathbb C$} & {-} & $\operatorname{Re} \rho \geq 0$ \\ \hline
					\end{tabular}\label{tab_summary}
			\end{table}
						
			\newpage
			\section*{Acknowledgements}{This work is supported by the Federal Ministry of Research, Technology and Space of Germany under Grant No. 05M22VSB and 05M22WHA.}
			%
			%

\bibliographystyle{unsrt}  
\bibliography{literature}   

			\newpage
			\appendix
            \section{Proof of Lemma \ref{lemma_forcing}}
            \label{appendix}
\begin{proof}[Proof of Lemma \ref{lemma_forcing}]
	Assume w.l.o.g. $t_0 = 0$. We first show that that $\mathcal F x_0 \in C[0,\infty)$. In other words, we must demonstrate that, for all $t \in [0, \infty)$, $\mathcal F {x_0}(t+\epsC) \to \mathcal F {x_0}(t)$ as $\epsC \to 0$.

To this end, let $t \ge 0$ and $\epsC > 0$ w.l.o.g. . Then,
\begin{equation}
	\mathcal F {x_0}(t +\epsC) - \mathcal F {x_0}(t) = \frac 1 {\Gamma(1-\alpha)} (B_1 + B_2)
\end{equation}
where
\begin{align}
	B_1 & = \int_{- \eta}^{0} 
					\left[ (t + \epsC - \tau)^{-\alpha} - (t - \tau)^{-\alpha} \right] {x_0}'(\tau) \, \mathrm d \tau 
\intertext{and}
	B_2 &=  \int_{-\infty}^{- \eta} 
			\left[ (t + \epsC - \tau)^{-\alpha} - (t - \tau)^{-\alpha} \right]  {x_0}'(\tau) \, \mathrm d \tau
\end{align}
where $\eta > 0$ is the fixed number from the property that ${x_0}' \in L^\infty [-\eta,0]$.
Recalling that ${x_0}'$ is bounded on $[-\eta, 0]$, we see that
\begin{align}\label{eq_boundUpperPart}
	|B_1| & \le \sup_{\tau \in [-\eta, 0]} |{x_0}'(\tau)| \cdot \int_{- \eta}^{0} 
					\left| (t + \epsC - \tau)^{-\alpha} - (t - \tau)^{-\alpha} \right|  \, \mathrm d \tau \\
		& = \sup_{\tau \in [-\eta, 0]} |{x_0}'(\tau)| \frac 1 {1-\alpha} \left[ (t + \eta)^{1-\alpha} - (t + \epsC  + \eta)^{1-\alpha}
						- t^{1-\alpha} + (t + \epsC )^{1-\alpha} 
						 \right],
\end{align}
where we used $ (t + \epsC - \tau)^{-\alpha} - (t - \tau)^{-\alpha} < 0$ and therefore $B_1 \to 0$ for $\epsC \to 0$.
Moreover, an integration by parts{, justified by the absolute continuity of $x_0$},  yields that
\begin{equation}
	B_2 = B_{21} - B_{22} 
\end{equation}
where
\begin{align}
	\label{eq:b21}
	B_{21} &= \left. \left[ (t + \epsC - \tau)^{-\alpha} - (t - \tau)^{-\alpha} \right]^{\vphantom{b}}
				\cdot {x_0}(\tau) \right |_{\tau = -\infty}^{\tau = - \eta}
\intertext{and}
	\nonumber
	B_{22} &= -\alpha \int_{-\infty}^{-\eta} \left[ (t+\epsC - \tau)^{-1-\alpha} - (t - \tau)^{-1-\alpha} \right] 
						{x_0}(\tau) \, \mathrm d \tau.
\end{align}
Clearly, since ${x_0}$ is bounded and the term in square brackets in \eqref{eq:b21} tends to zero as $\tau \to -\infty$, we find that
\begin{equation}
	B_{21} = \left[ (t + \epsC  + \eta)^{-\alpha} - (t + \eta)^{-\alpha} \right]
				\cdot {x_0}(- \eta)
		\to 0
\end{equation}
for $\epsC \to 0$ since $t \ge 0$ and $\eta > 0$. Finally, once again using the boundedness of ${x_0}$ and applying the
mean value theorem of differential calculus to the function $(\cdot)^{-1-\alpha}$ on the interval $[t - \tau, t+\epsC - \tau]$, we see that
\begin{align}
	|B_{22}| 
	& \le \alpha \sup_{\tau \le - \eta} |{x_0}(\tau)| 
				 \int_{-\infty}^{-\eta} \left| (t+\epsC - \tau)^{-1-\alpha} - (t - \tau)^{-1-\alpha} \right| 
				 	\, \mathrm d \tau \\
	&= \alpha (\alpha + 1) \sup_{\tau \le - \eta} |{x_0}(\tau)| \cdot \epsC  
				 \int_{-\infty}^{-\eta} \left| \xi_{\tau, \epsC} \right|^{-2-\alpha} 
				 	\, \mathrm d \tau
\end{align}
where $\xi_{\tau, \epsC} \in [t - \tau, t+\epsC - \tau]$ and thus $\xi_{\tau, \epsC} \geq  t - \tau$. Therefore,
\begin{align}
	|B_{22}| & \le \epsC ~ \alpha (\alpha + 1) \sup_{\tau \le - \eta} |{x_0}(\tau)| \cdot 
			 \int_{-\infty}^{-\eta} (t - \tau)^{-2-\alpha} 
				 	\, \mathrm d \tau 
\end{align}
and so $B_{22} \to 0$ for $\epsC \to 0$ as well because the relations $\tau < - \eta < 0 \le t$ imply that
the integral in the last equation is finite. We thus have shown continuity.
    
    We next show the bound \eqref{eq_forcingBoundInftyNormNew}. First, we split the forcing term into the two integrals
	\begin{align}
		\mathcal Fx_0(t)\Gamma(1-\alpha) = I_1(t)+I_2(t),\\
		I_1(t)=\int_{-\eta}^{0}(t-\tau)^{-\alpha}x_0'(\tau)\dsomething{\tau},\quad I_2(t)=\int_{-\infty}^{-\eta}(t-\tau)^{-\alpha}x_0'(\tau)\dsomething{\tau},
	\end{align} where $[-\eta,0]$ is the interval on which $x_0'$ is bounded according to Definition \ref{def_InitialSpace}. We derive for the first integral the bound
	\begin{align}
		||I_1(t)||&\leq \int_{-\eta}^{0}(t-\tau)^{-\alpha}||x_0'(\tau)||\dsomething{\tau}
		\leq \left(\sup_{t\in [-\eta,0]}||x_0'(t)||\right)\int_{-\eta}^{0}(t-\tau)^{-\alpha}\dsomething{\tau}\\
		&=  \left(\sup_{t\in [-\eta,0]}||x_0'(t)||\right) \frac{1}{1-\alpha}\left((t+\eta)^{1-\alpha}-t^{1-\alpha}\right)
		\leq  \left(\sup_{t\in [-\eta,0]}||x_0'(t)||\right) \frac{\eta}{1-\alpha}(t+\eta)^{-\alpha}.\label{eq_I1bound}
	\end{align}
	The last step holds by the following argument. Let $\alpha\in(0,1)$ and $t\geq 0$. We want to show that
	\begin{align}
		(t+\eta)^{1-\alpha}-t^{1-\alpha}&\leq \eta(t+\eta)^{-\alpha}.\label{eq_AlgebraicBoundStart}
	\end{align}Dividing by the first term yields
	\begin{align}
		1-\left(\frac{t}{t+\eta}\right)^{1-\alpha} &\leq \frac{\eta}{t+\eta}. \label{eq_AlgebraicBound}
	\end{align} Introducing $\zeta = \frac{t}{t+\eta}\in [0,1)$ and noting $1-\zeta = \frac{\eta}{t+\eta}$, equation \eqref{eq_AlgebraicBound} may be reformulated to
	\begin{align}
		\zeta^{1-\alpha} \geq \zeta
	\end{align}which does hold since $\zeta\in[0,1)$. Thus, equation \eqref{eq_AlgebraicBoundStart} indeed holds. For the second integral we apply  integration by parts{, justified by the absolute continuity of $x_0$}, 
	\begin{align}\label{eq_intByParts}
		I_2(t) &= \int_{-\infty}^{-\eta}(t-\tau)^{-\alpha}x_0'(\tau)\dsomething{\tau}= \left[(t-\tau)^{-\alpha} x_0(\tau)\right]^{\tau = -\eta}_{\tau = -\infty}+\alpha\int_{-\infty}^{-\eta}(t-\tau)^{-\alpha-1}x_0(\tau)\dsomething{\tau}\\&= (t+\eta)^{-\alpha}x_0(-\eta)+\alpha\int_{-\infty}^{-\eta}(t-\tau)^{-\alpha-1}x_0(\tau)\dsomething{\tau}\label{eq_I2BoundIntermediate}
	\end{align} where the limit $\lim_{\tau\to-\infty}x_0(\tau)/(t-\tau)^\alpha$ vanishes for all $t\geq0$ since $||x_0||_\infty\leq \infty$. The absolute value of the integral may be bounded as follows:
	\begin{align}
		\alpha\left\lVert\int_{-\infty}^{-\eta}(t-\tau)^{-\alpha-1}x_0(\tau)\dsomething{\tau}\right\rVert\leq \alpha ||x_0||_\infty \int_{-\infty}^{-\eta}(t-\tau)^{-\alpha-1}\dsomething{\tau}=||x_0||_\infty (t+\eta)^{-\alpha} \label{eq_I3bound}
	\end{align}
	Using equations \eqref{eq_I2BoundIntermediate} and \eqref{eq_I3bound} and $|x_0(-\eta)|\leq||x_0||_\infty$, we may set up the bound
	\begin{align}
		||I_2(t)||\leq 2||x_0||_\infty (t+\eta)^{-\alpha}.\label{eq_I2bound}
	\end{align}From the triangle inequality $||\mathcal Fx_0(t)|| \Gamma(1-\alpha)\leq ||I_1(t)||+||I_2(t)||$ and the established bounds \eqref{eq_I1bound} and \eqref{eq_I2bound}, one arrives at the proposed \eqref{eq_forcingBoundInftyNormNew} and the proof is complete.
\end{proof}

		\end{document}